\documentclass[11pt,reqno]{amsart}

\usepackage{graphicx}
\usepackage{amsmath}
\usepackage{amsmath}
\usepackage{amsfonts}
\usepackage{amssymb}
\usepackage{algorithm,algorithmic}
\usepackage{multirow}
\usepackage{txfonts}
\usepackage{verbatim}
\usepackage{subfig}
\usepackage{setspace}
\def\@xthm#1#2{\@begintheorem{#2}{\csname the#1\endcsname}{}\ignorespaces}
\def\@ythm#1#2[#3]{\@opargbegintheorem{#2}{\csname
       the#1\endcsname}{#3}\ignorespaces}

\def\@begintheorem#1#2#3{\par\addvspace{8pt plus3pt minus2pt}%
              \noindent{\csname#1headfont\endcsname#1\ \ignorespaces#3 #2.}%
              \csname#1font\endcsname\hskip6pt\ignorespaces}
\def\@endtheorem{\par\addvspace{8pt plus3pt minus2pt}\@endparenv}
\def\@opargbegintheorem#1#2#3{\par\addvspace{6pt plus3pt minus2pt}%
	\def\@tempa{#3}%
	\noindent{\bf #1 #2 \ifx\@tempa\empty\unskip\else\unskip\
(#3).\fi\hskip.5em}\csname#1font\endcsname\ignorespaces
#3\fi\hskip1em}\it

\def\@endtheorem{\par\addvspace{6pt plus3pt minus2pt}}
\newtheorem{theorem}{Theorem}[section]
\newtheorem{corollary}{Corollary}[section]
\newtheorem{lemma}{Lemma}[section]

\newtheorem{remark}{Remark}[section]

\newif\iflogo
\def\prbox{\par
	\vskip-\lastskip\vskip-\baselineskip\hbox to
\hsize{\hfill\fboxsep0pt\fbox{\phantom{\vrule width5pt height5pt
depth0pt}}}\global\logofalse}
	
\begin{document}

\author{Fr\'ed\'eric Abergel and Aymen Jedidi}
\address{Chaire de Finance
Quantitative, Laboratoire de Math\'ematiques Appliqu\'ees aux Syst\`emes, \'{E}cole Centrale Paris, 92290 Ch\^{a}tenay-Malabry, France.}
\email{frederic.abergel@ecp.fr, aymen.jedidi@ecp.fr.}

\date{November, 2012}

\title{\textbf{A Mathematical Approach to Order Book Modeling}}

\maketitle


\begin{abstract}
Motivated by the desire to bridge the gap between the microscopic
description of price formation (agent-based modeling) and the stochastic
differential equations approach used classically to describe price evolution at
macroscopic time scales, we present a mathematical study of the order book as a multidimensional continuous-time Markov chain
and derive several mathematical results in the case of
independent Poissonian arrival times. In particular, we show that the cancellation
structure is an important factor ensuring the existence of a stationary distribution
and the exponential convergence towards it. We also prove, by means of the functional
central limit theorem (FCLT), that the rescaled-centered price process converges
to a Brownian motion. We illustrate the analysis with numerical simulation and comparison against market data.
\end{abstract}

\section{Introduction and Background}

The emergence of electronic trading as a major means of trading financial assets
makes the study of the order book central to understanding the mechanisms of
price formation. In order-driven markets, buy and sell orders are matched
continuously subject to price and time priority. The \emph{order book} is the
list of all buy and sell limit orders, with their corresponding price and size,
at a given instant of time. Essentially, three types of orders can be submitted:
    \begin{itemize}
      \item \emph{Limit order}: Specify a price (also called ``quote'') at which
one is willing to buy or sell a certain number of shares;
      \item \emph{Market order}: Immediately buy or sell a certain number of
shares at the best available opposite quote;
      \item \emph{Cancellation order}: Cancel an existing limit order.
    \end{itemize}
In the 
literature, ``agents'' who submit exclusively limit orders
are referred to as \emph{liquidity providers}. Those who submit market orders
are referred to as \emph{liquidity takers}.

Limit orders are stored in the order book until they are either executed against
an incoming market order or canceled. The \emph{ask} price $P^A$ (or simply the
ask) is the price of the best (i.e. lowest) limit sell order. The \emph{bid}
price $P^B$ is the price of the best (i.e. highest) limit buy order. The gap
between the bid and the ask
\begin{equation}
S := P^A - P^B,
\end{equation}
is always positive and is called the \emph{spread}. Prices are not continuous,
but rather have a discrete resolution $\Delta P$, the \emph{tick}, which
represents the smallest quantity by which they can change. We define the
\emph{mid-price} as the average between the bid and the ask
\begin{equation}
P:=\frac{P^A+P^B}{2}.
\end{equation}

The price dynamics is the result of the interplay between the incoming order
flow and the order book \cite{Bouchaud}. Figure \ref{fig1} is a schematic
illustration of this process \cite{Ferraris}. Note that we chose to represent
quantities on the bid side of the book by non-positive numbers.

\begin{figure}
    \begin{center}
        \includegraphics[width=\textwidth]{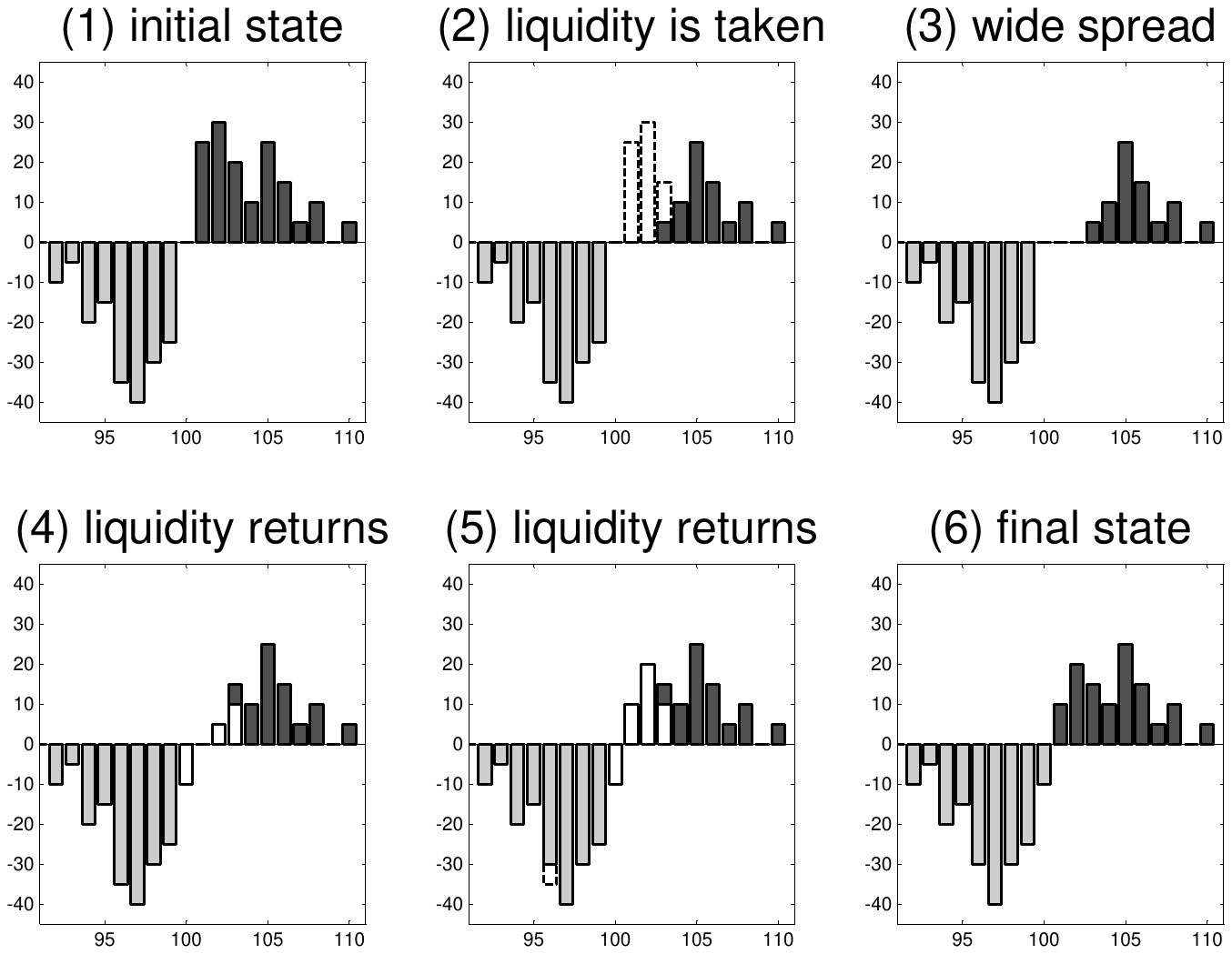}
        \caption{Order book schematic illustration: a buy market order
arrives and removes liquidity from the ask side,
        then sell limit orders are submitted and liquidity is restored.}
        \label{fig1}
    \end{center}
\end{figure}

Although in reality orders can have any size, we shall assume in most of the paper that all orders have a fixed unit size $q$. This assumption is convenient to carry out our analysis and is, for now, of secondary importance to the
problem we are interested in\footnote{It will be relaxed in section \ref{NumericalExample} where we resort to numerical simulation.}. Throughout the paper, we may refer to three different ``times'':
    \begin{itemize}
        \item \emph{Physical time} (or \emph{clock time}) in seconds,
        \item \emph{Event time} (or \emph{tick time}): the time counter is incremented by 1 every time an event (i.e. market, limit or cancellation order) occurs,
        \item \emph{Trade time} (or \emph{transaction time}): the time counter is incremented every trade (i.e. every market order).
    \end{itemize}

\smallskip \noindent \textbf{Related literature.} Order book modeling has been an area of intense research activity in the last decade. The remarkable interest in this area is due to two factors:

\begin{itemize}
    \item Widespread use of algorithmic trading in which the order book is the place where offer and demand meet,
    \item Availability of tick by tick data that record every change in the order book and allow precise analysis of the price formation process at the microscopic level.
\end{itemize}

Schematically, two modeling approaches have been successful in capturing key properties of the order book---at least partially. The first approach, led by economists, models the interactions between rational agents who act strategically: they choose their trading decisions as solutions to individual utility maximization problems (See e.g. \cite{ParlourSurvey2008} and references therein).

In the second approach, proposed by econophysicists, agents are assumed to act randomly. 
This is sometime referred to as zero-intelligence order book modeling, in the sense that order arrivals and placements are entirely stochastic. 
The focus here is more on the ``mechanistic'' aspects of the continuous double auction rather than the strategic interactions between agents.
Despite this apparent limitation, zero-intelligence (or statistical) 
order book models do capture many salient features of real markets (See \cite{DanielsFarmer, FarmerPatelli} and references therein).
Two notable developments in this strand of research are \cite{Maslov} 
who proposed one of the earliest stochastic order book models,
and \cite{ChalletStinchcombe1} who added the possibility to cancel existing limit orders.

In their seminal paper \cite{SmithFarmer}, Smith et al. develop a dynamical statistical order book model under the assumption of IID Poissonian order flow. They provide a thorough analysis of the model using simulation, dimensional analysis and mean field approximation. They study key characteristics of the model, namely:
\begin{enumerate}
    \item Price diffusion.
    \item Liquidity characteristics: average depth profile, bid-ask spread, price impact and time and probability to fill a limit order.
\end{enumerate}

One of the most important messages of their analysis is that zero-intelligence order book models are able to produce reasonable market dynamics and liquidity characteristics. Our focus here is on the first point, that is, \emph{the convergence of the price process, which is a jump process at the microscopic level, to a diffusive process\footnote{In this paper, we mean (abusively) by ``diffusive process'' or simply ``diffusion'' the mathematical concept of Brownian motion.} at macroscopic time scales}. The authors in \cite{SmithFarmer} suggest that a diffusive regime is reached. Their argument relies on a mean field approximation. Essentially, this amounts to neglecting the dependence between order fluctuations at adjacent price levels.

Another important paper of interest to us is \cite{ContStoikovTalreja}. Cont et al. propose to model the order book dynamics from the vantage point of queuing systems. They remarkably succeed in deriving many conditional probabilities of practical importance such as the probability of an increase in the mid-price, of the execution of an order at the bid before the ask quote moves, and of ``making the spread''. To our knowledge, they are the first to clearly set the problem of stochastic order book modeling in the context of Markov chains, which is a very powerful and well-studied mathematical concept.

\smallskip \noindent \textbf{Outline.} In this paper, we build on the models of \cite{ContStoikovTalreja} and \cite{SmithFarmer} to present a stylized description of the order book,
and derive several mathematical results in the case of independent Poissonian arrival times. In particular, we show that the cancellation structure is an important factor ensuring the existence of a stationary distribution for the order book and the exponential convergence towards it. We also prove, by means of the functional central limit theorem (FCLT), that the rescaled-centered price process converges to a Brownian motion, which is a new result.

The remainder of the paper is organized as follows. In section \ref{PerfectMarketMaking}, we motivate our approach using an elementary example where the spread is kept constant (``perfect market making''). In sections \ref{OrderBookDynamics} trough \ref{PriceDynamics}, we compute the infinitesimal generator associated with the order book in a general setting, and link the price dynamics to the instantaneous state of the order book. In section \ref{ErgodicityAndDiffusiveLimit}, we prove that the order book is \emph{ergodic}---in particular it has a \emph{stationary distribution}---that it converges to its stationary state \emph{exponentially fast}, and that the large-scale limit of the price process is a \emph{Brownian motion}. Our proofs rely on the theory of infinitesimal generators and Foster-Lyapunov stability criteria for Markov chains. We outline an order book simulation algorithm in section \ref{NumericalExample} and provide a numerical illustration. Finally, section \ref{Conclusion} summarizes our 
results and contains critiques of Markovian order book models.

\section{An Elementary Approximation: Perfect Market Making}

\label{PerfectMarketMaking}
We start with the simplest agent-based market model:
\begin{itemize}
    \item The order book starts in a full state: All limits above $P^A(0)$
and below $P^B(0)$ are filled with one limit order of unit size $q$. The spread
starts equal to $1$ tick;
    \item The flow of market orders is modeled by two independent Poisson
processes $M^+(t)$ (buy orders) and $M^-(t)$ (sell orders) with constant arrival
rates (or intensities) $\lambda^+$ and $\lambda^-$;
    \item There is one liquidity provider, who reacts \emph{immediately} after a
market order arrives so as to maintain the spread constantly equal to $1$ tick. He
places a limit order on the same side as the market order (i.e. a buy limit
order after a buy market order and vice versa) with probability $u$ and on the
opposite side with probability $1-u$.
\end{itemize}
The mid-price dynamics can be written in the following form
\begin{align}
dP(t) &= \Delta P \; (dM^{+}(t) - dM^{-}(t))Z,
\end{align}
where $Z$ is a Bernoulli random variable
\begin{equation}
Z = 0 \text{ with probability } (1-u),
\end{equation}
and
\begin{equation}
Z = 1 \text{ with probability } u.
\end{equation}
The infinitesimal generator\footnote{The infinitesimal generator of a
time-homogeneous Markov process $(\mathbf{X}(t))_{t\geq 0}$ is the operator
$\mathcal{L}$, if exists, defined to act on sufficiently regular functions
$f:\mathbb{R}^n \rightarrow \mathbb{R}$, by
\begin{equation}
\mathcal{L} f (\mathbf{x}) := \lim_{t \downarrow 0} \cfrac{\mathbb{E}
[f(\mathbf{X}(t))|\mathbf{X}(0)=\mathbf{x}] - f(\mathbf{x})}{t}.
\end{equation}
It provides an analytical tool to study $(\mathbf{X}(t)).$
} $\mathcal{L}$ associated with this dynamics is
\begin{equation}
\mathcal{L} f (P) = u \; \left[ \lambda^+ \; (f(P+\Delta P) - f) + \lambda^- \;
(f(P-\Delta P) - f)  \right],
\label{generator}
\end{equation}
where $f$ denotes a test function.
It is well known that a continuous limit is obtained under suitable assumptions
on the intensity and tick size. Noting that \eqref{generator} can be rewritten
as
\begin{eqnarray}
\label{genrator2}
\mathcal{L} f (P) &=& \frac{1}{2} u \; (\lambda^+ + \lambda^-) (\Delta P) ^2
\frac{f(P+\Delta P) - 2 f + f(P-\Delta P)}{(\Delta P)^2} \notag\\
&+& u \; (\lambda^+-\lambda^-) \Delta P \frac{f(P+\Delta P)-f(P-\Delta P)}{2 \Delta
P},
\end{eqnarray}
and under the following assumptions
\begin{equation}
u \; (\lambda^+ + \lambda^-) (\Delta P)^2 {\longrightarrow} \sigma^2 \; \text{ as }
\Delta P \rightarrow 0,
\end{equation}
and
\begin{equation}
u \; (\lambda^+ - \lambda^-) \Delta P {\longrightarrow}  \mu \; \text{ as }\Delta P
\rightarrow 0,
\end{equation}
the generator converges to the classical diffusion operator
\begin{equation}
\frac{\sigma^2}{2}\frac{\partial^2 f}{\partial P^2} + \mu \frac{\partial
f}{\partial P},
\end{equation}
corresponding to a Brownian motion with drift. This simple case is worked out as
an example of the type of limit theorems that we will be interested in in the
sequel. One should also note that a more classical approach using the Functional
Central limit Theorem (FCLT) as in \cite{Billingsley2} or \cite{Whitt} yields
similar results ; For given fixed values of $\lambda^+$, $\lambda^-$ and $\Delta
P$, the rescaled-centred price process
\begin{equation}
\frac{ P(n t)- n \mu t}{\sqrt{n} \sigma}
\end{equation}
converges as
$n \rightarrow \infty$,
to a standard Brownian motion $(B(t))$ where
\begin{equation}
\sigma = \Delta P \sqrt{(\lambda^+ + \lambda^-) u},
\end{equation}
and
\begin{equation}
\mu = \Delta P (\lambda^+ - \lambda^-) u.
\end{equation}
Let us also mention that one can easily achieve more complex diffusive limits such as
a local volatility model by imposing that the limit is a function of $P$ and $t$
\begin{equation}
u \; (\lambda^+ + \lambda^-) (\Delta P) ^2 \rightarrow \sigma^2(P,t),
\end{equation}
and
\begin{equation}
u \; (\lambda^+ - \lambda^-) \Delta P \rightarrow \mu(P,t).
\end{equation}
This is the case if the original intensities are functions of $P$ and $t$
themselves.

\section{Order Book Dynamics}

\label{OrderBookDynamics}

\subsection{Model setup: Poissonian arrivals, reference frame and boundary
conditions}

We now consider the dynamics of a general order book under the assumption of Poissonian arrival times
for market orders, limit orders and cancellations.
We shall assume that each side of the order book is fully described by a
\emph{finite} number of limits $K$, ranging from $1$ to $K$ ticks away from the
best available opposite quote. We will use the notation\footnote{In what follows, bold notation indicates vector quantities.}
\begin{equation}
\mathbf{X}(t) := (\mathbf{a}(t); \mathbf{b}(t)) := \left( a_1(t),\dots,a_K(t) ;
b_1(t),\dots,b_K(t) \right ),
\end{equation}
where
$\mathbf{a} :=  \left( a_1,\dots,a_K\right )$
designates the ask side of the order book and $a_i$ the number of shares
available $i$ ticks away from the best opposite quote,
and
$\mathbf{b} :=  \left( b_1,\dots,b_K\right )$
designates the bid side of the book. By doing so, we adopt the representation
described in \cite{ContStoikovTalreja} or \cite{SmithFarmer}\footnote{See
also \cite{GatheralOomen} for an interesting discussion.}, but depart slightly
from it by adopting a \emph{finite moving frame}, as we think it is realistic
and more convenient when scaling in tick size will be addressed.

Let us now recall the events that may happen:
\begin{itemize}
    \item arrival of a new market order;
    \item arrival of a new limit order;
    \item cancellation of an already existing limit order.
\end{itemize}
These events are described by \emph{independent} Poisson processes:
\begin{itemize}
    \item $M^{\pm}(t)$: arrival of new market order, with intensity
$\displaystyle \lambda^{M^+} \mathbb{I}({\mathbf{a}\neq\mathbf{0}})$ and
$\lambda^{M^-} \mathbb{I}({\mathbf{b}\neq\mathbf{0}})$;
    \item $L^{\pm}_i(t)$: arrival of a limit order at level $i$, with
intensity $\displaystyle \lambda^{L^{\pm}}_i$;
    \item $C^{\pm}_i(t)$: cancellation of a limit order at level $i$, with
intensity $\displaystyle \lambda^{C^+}_i {a_i} $ and $\displaystyle
\lambda^{C^-}_i |b_i| $.
\end{itemize}
$q$ is the size of any new incoming order, and the superscript ``$+$''
(respectively ``$-$'') refers to the ask (respectively bid) side of the book.
Note that the intensity of the cancellation process at level $i$ is \emph{proportional}
to the available quantity at that level. That is to say, each order at level $i$
has a lifetime drawn from an exponential distribution with intensity
$\lambda^{C^{\pm}}_i$. Note also that buy limit orders $L^-_i(t)$ arrive below
the ask price $P^{A}(t)$, and sell limit orders $L^+_i(t)$ arrive above the bid
price $P^B(t)$.

We impose constant boundary conditions outside the moving frame of size $2K$:
Every time the moving frame leaves a price level, the number of shares at that
level is set to $a_{\infty}$ (or $b_{\infty}$ depending on the side of the
book).
Our choice of a finite moving frame and constant\footnote{Actually, taking for
$a_{\infty}$ and $|b_{\infty}|$ independent positive random variables would not
change much our analysis. We take constants for simplicity.} boundary conditions
has three motivations. Firstly, it assures that the order book does not empty
and that $P^A$, $P^B$ are always well defined. Secondly, it keeps the spread $S$
and the increments of $P^A$, $P^B$ and $P=(P^A+P^B)/2$ bounded---This will
be important when addressing the scaling limit of the price. Thirdly, it makes
the model Markovian as we do not keep track of the price levels that have been
visited (then left) by the moving frame at some prior time. Figure
\ref{dynamics} is a representation of the order book using the above notations.

\begin{figure}
\begin{center}
\includegraphics[width=0.8\textwidth]{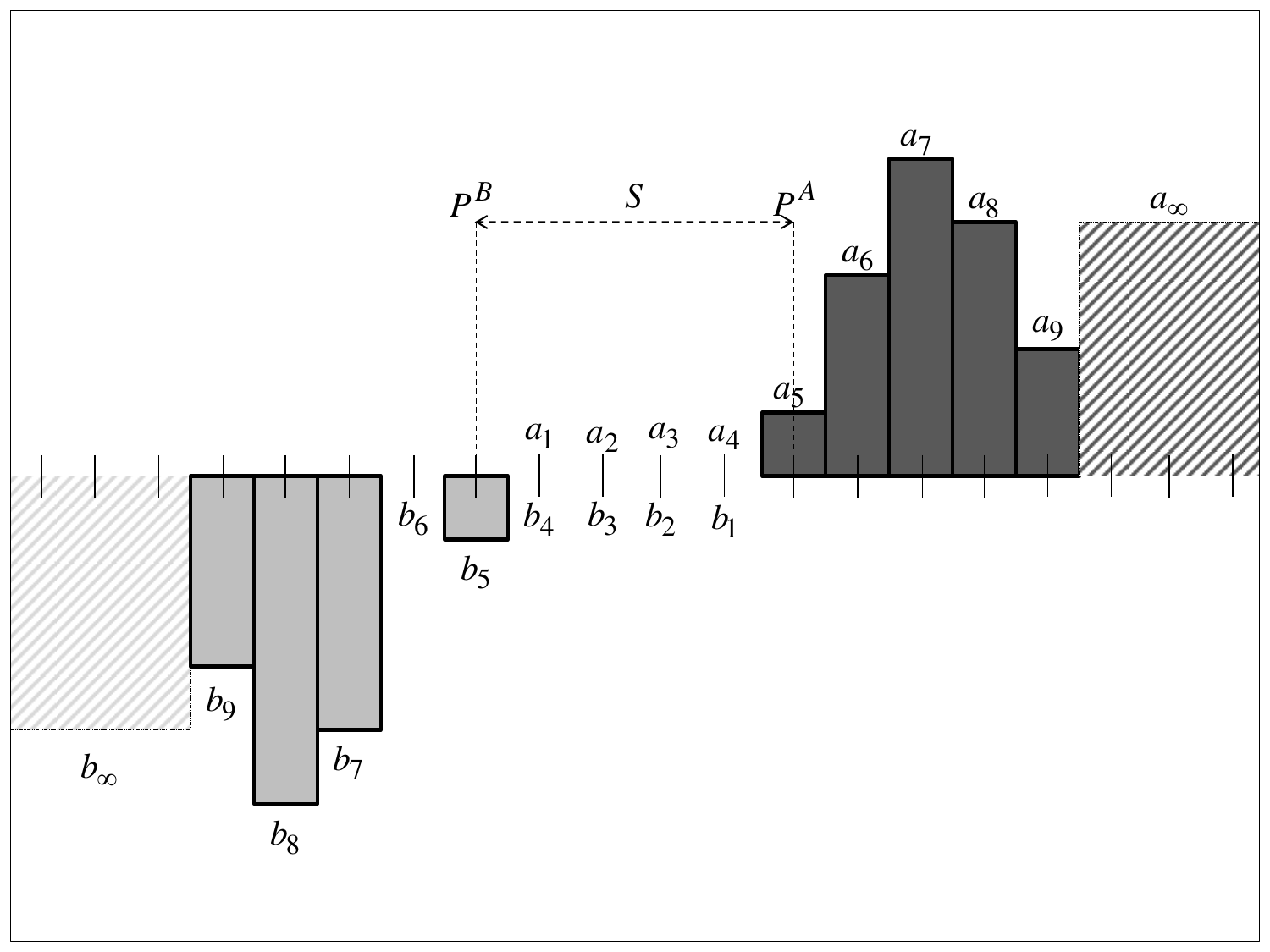}
\end{center}
\caption{\label{dynamics}\small Order book dynamics: in this example, $K=9$, $q = 1$,
$a_{\infty} = 4$, $b_{\infty}=-4$. The shape of the order book is such that
$\mathbf{a}(t)=(0,0,0,0,1,3,5,4,2)$ and $\mathbf{b}(t)=(0,0,0,0,-1,0,-4,-5,-3)$.
The spread $S(t)=5$ ticks. Assume that at time $t'>t$ a sell market order
$dM^-(t')$ arrives, then $\mathbf{a}(t')=(0,0,0,0,0,0,1,3,5)$,
$\mathbf{b}(t')=(0,0,0,0,0,0,-4,-5,-3)$ and $S(t')=7$. Assume instead that at
$t'>t$ a buy limit order $dL_1^-(t')$ arrives one tick away from the best
opposite quote, then $\mathbf{a}(t')=(1,3,5,4,2,4,4,4,4)$,
$\mathbf{b}(t')=(-1,0,0,0,-1,0,-4,-5,-3)$ and $S(t') = 1$.}
\end{figure}

\subsection{Comparison to previous results and models}
Before we proceed, we would like to recall some results already present in the literature and highlight their differences with respect to our analysis. Smith et al. have already investigated in \cite{SmithFarmer} the scaling properties of some liquidity and price characteristics in a stochastic order book model. These results are summarized in table \ref{SmithResults}. In the model of Smith et al. \cite{SmithFarmer}, orders arrive on an \emph{infinite} price grid (This is consistent as limit orders arrival rate \emph{per price level} is finite). Moreover, the arrival rates are independent of the price level, which has the advantage of enabling the analytical predictions summarized in table \ref{SmithResults}.
\begin{table}[H]
    \begin{center}
        \begin{tabular}{|c|c|}
                \hline
                Quantity & Scaling relation \\
                \hline
                Average asymptotic depth & ${\lambda^L}/{\lambda^C}$\\
                \hline
                Average spread & ${\lambda^M}/{\lambda^L} f(\epsilon, {\Delta P}/{p_c})$ \\
                \hline
                Slope of average depth profile & ${(\lambda^L)^2}/{\lambda^M \lambda^C} g(\epsilon, {\Delta P}/{p_c})$ \\
                \hline
                Price ``diffusion'' parameter at short time scales & ${(\lambda^M)^2 \lambda^C}/{\lambda^L} \epsilon^{-0.5}$\\
                \hline
                Price ``diffusion'' parameter at long time scales & ${(\lambda^M)^2 \lambda^C}/{\lambda^L} \epsilon^{0.5}$ \\
                \hline
                \end{tabular}
        \caption{\label{SmithResults} Results of Smith et al. $\displaystyle \epsilon:={q}/{({\lambda^M}/{2 \lambda^C})}$ is a ``granularity'' parameter that characterizes the effect of discreteness in order sizes, $\displaystyle p_c:={\lambda^M}/{2 \lambda^L}$ is a characteristic price interval, and $f$ and $g$ are slowly varying functions.} 
    \end{center}
\end{table}
We stress that, to our understanding, these results are obtained by mean-field approximations, which assume that the fluctuations at adjacent price levels are independent. This allows fruitful simplifications of the complex dynamics of the order book. In addition, the authors do not characterize the convergence of the coarse-grained price process in the sense of Stochastic Process Limits, nor do they show that the limiting process is precisely a Brownian motion (theorem \ref{MainResult}).

In the model of Cont el al. \cite{ContStoikovTalreja}, arrival rates are indexed by the distance to the best opposite quote, which is more realistic. The order book is constrained to a finite price grid $[1, P_{max}]$ that facilitates the analysis of the Markov chain. Here, we use a combination of the two models in that the arrival rates are not uniformly distributed across prices, and the reference frame is finite but moving. 
Cont et al. \cite{ContStoikovTalreja} have considered the question of the ergodicity of their order book model. We also address this question following a different route, and more importantly to our analysis, exhibit the \emph{rate of convergence} to the stationary state, which turns out to be the key of the proof of theorem \ref{MainResult}.

\subsection{Evolution of the order book}

We can write the following coupled SDEs for the quantities of outstanding
limit orders in each side of the order book:\footnote{Remember that, by
convention, the $b_i$'s are non-positive.}
\begin{eqnarray}
da_i(t) & = & -\left(q - \sum_{k=1}^{i-1}{a_k} \right)_+ dM^+(t) + q
dL^+_i(t) - q dC^+_i(t) \nonumber \\
 & + & (J^{M^-}(\mathbf{a})-\mathbf{a})_i dM^-(t)  +
\sum_{i=1}^K{(J^{L_i^-} (\mathbf{a})-\mathbf{a})_i dL_i^-(t)} \nonumber \\
&+&\sum_{i=1}^K{(J^{C_i^-}(\mathbf{a})-\mathbf{a})_i dC_i^-(t)}, 
\label{obdynamics}
\end{eqnarray}
and
\begin{eqnarray}
db_i(t) &  = & \left(q - \sum_{k=1}^{i-1}{|b_k|} \right)_+ dM^-(t) -
q dL_i^-(t) + q dC_i^-(t) \nonumber \\
 &  + & (J^{M^+}(\mathbf{b})-\mathbf{b})_i dM^+(t) +
\sum_{i=1}^K{(J^{L_i^+}(\mathbf{b})-\mathbf{b})_i dL_i^+(t)} \nonumber \\
& + & \sum_{i=1}^K{(J^{C_i^+}(\mathbf{b})-\mathbf{b})_i dC_i^+(t)},
\end{eqnarray}
where the $J$'s are \emph{shift operators} corresponding to the renumbering of
the ask side following an event affecting the bid side  of the book and vice
versa. For instance the shift operator corresponding to the arrival of a sell
market order $dM^-(t)$ of size $q$ is\footnote{For notational simplicity, we
write $J^{M^-}(\mathbf{a})$ instead of $J^{M^-}(\mathbf{a};\mathbf{b})$ etc. for
the shift operators.}
\begin{equation}
J^{M^-}(\mathbf{a}) = \left(\underbrace{0, 0, \dots, 0}_{\text{k times}}, a_1,
a_2, \dots, a_{K-k} \right),
\end{equation}
with
\begin{equation}
k := \inf\{p : \sum_{j=1}^{p}{|b_j|}>q\} - \inf\{p : {|b_p|}>0\}.
\end{equation}
Similar expressions can be derived for the other events affecting the order
book.

In the next sections, we will study some general properties of the order book,
starting with the generator associated with this $2K$-dimensional
continuous-time Markov chain.

\section{Infinitesimal Generator}

\label{InfinitesimalGenerator}

Let us work out the infinitesimal generator associated with the jump process
described above. We have
\begin{align}
\mathcal{L} f \left(\mathbf{a};\mathbf{b}\right) &=\lambda^{M^+} ( f\left( [a_i
- (q - A({i-1}) )_+]_+; J^{M^+}(\mathbf{b})\right) -f )\notag\\
&+ \sum_{i=1}^{K}{\lambda_i^{L^+}  (f\left(a_i+q; J^{L_i^+}(\textbf{b})\right) -
f)}\notag\\
&+ \sum_{i=1}^{K}{\lambda_i^{C^+} a_i (f\left(a_i-q;
J^{C_i^{+}}(\textbf{b})\right) - f)}\notag\\
& + \lambda^{M^-} {\left( f\left(J^{M^-}(\mathbf{a}); [b_i + (q - B({i-1})
)_+]_-\right) -f \right)}\notag\\
&+ \sum_{i=1}^{K}{\lambda_i^{L^-} (f\left(J^{L_i^-}(\textbf{a}); b_i-q\right) -
f)}\notag\\
&+  \sum_{i=1}^{K}{\lambda_i^{C^-} |b_i|
(f\left(J^{C_i^-}(\textbf{a}); b_i+q\right) - f)},
\label{infgen}
\end{align}
where, to ease the notations, we note $f(a_i; \mathbf{b})$ instead of
$f(a_1,\dots,a_i,\dots,a_K;\mathbf{b})$ etc. and
\begin{equation}
x_+:=\max(x,0), \qquad x_-:=\min(x,0), x \in \mathbb{R}.
\end{equation}
The operator above, although cumbersome to put in writing, is simple to
decipher: a series of standard difference operators corresponding to the
``deposition-evaporation'' of orders at each limit, combined with the shift
operators expressing the moves in the best limits and therefore, in the origins
of the frames for the two sides of the order book. Note the coupling of the two
sides: the shifts on the $a$'s depend on the $b$'s, and vice versa. More
precisely the shifts depend on the profile of the order book on the other side,
namely the cumulative depth up to level $i$ defined by
\begin{equation}
A(i) := \sum_{k=1}^i{a_k},
\end{equation}
and
\begin{equation}
B(i) := \sum_{k=1}^i{|b_k|},
\end{equation}
and the generalized inverse functions thereof
\begin{equation}
A^{-1}(q') := \inf\{p: \sum_{j=1}^{p} a_j > q' \},
\end{equation}
and
\begin{equation}
B^{-1}(q') := \inf\{p: \sum_{j=1}^{p} |b_j| > q' \},
\end{equation}
where $q'$ designates a certain quantity of shares\footnote{Note that a more rigorous
notation would be $$A(i,\mathbf{a}(t)) \text{ and } A^{-1}(q',\mathbf{a}(t))$$
for the depth and inverse depth functions respectively. We drop the dependence
on the last variable as it is clear from the context.}.
\begin{remark}
The index corresponding to the best opposite quote equals the spread $S$ in
ticks, that is
\begin{equation}
i_A := A^{-1}(0) = \inf\{p: \sum_{j=1}^{p} a_j > 0 \}  = \frac{S}{\Delta
P}:=i_S,
\end{equation}
and
\begin{equation}
i_B := B^{-1}(0) = \inf\{p: \sum_{j=1}^{p} |b_j| > 0 \} =\frac{S}{\Delta
P}:=i_S = i_A.
\end{equation}
\end{remark}

\section{Price Dynamics}

\label{PriceDynamics}

We now focus on the dynamics of the best ask and bid prices, denoted by $P^A(t)$
and $P^B(t)$. One can easily see that they satisfy the following SDEs:
\begin{eqnarray}
dP^A(t) &=& \Delta P [ (A^{-1}(q) - A^{-1}(0)) dM^+(t)  \nonumber \\
 &-& \sum_{i=1}^K{(A^{-1}(0)-i)_+ dL_i^+(t)} +
(A^{-1}(q)-A^{-1}(0))dC_{i_A}^+(t) ],
\end{eqnarray}
and
\begin{eqnarray}
dP^B(t) &=& - \Delta P [ (B^{-1}(q) - B^{-1}(0)) dM^-(t) \nonumber \\
 &-& \sum_{i=1}^K{(B^{-1}(0)-i)_+ dL_i^{-}(t)} +
(B^{-1}(q)-B^{-1}(0))dC_{i_B}^-(t) ],
\end{eqnarray}
which describe the various events that affect them: change due to a market
order, change due to limit orders inside the spread, and change due to the
cancellation of a limit order at the best price. Equivalently, the respective
dynamics of the mid-price and the spread are:
\begin{align}
dP(t) &= \frac{\Delta P}{2} \left[  (A^{-1}(q) - A^{-1}(0)) dM^+(t) - (B^{-1}(q)
- B^{-1}(0)) dM^-(t) \right.\notag\\
& - \sum_{i=1}^K{(A^{-1}(0)-i)_+ dL_i^+(t)} + \sum_{i=1}^K{(B^{-1}(0)-i)_+
dL_i^-(t)} \notag\\
& + \left.  (A^{-1}(q) - A^{-1}(0)) dC_{i_A}^+(t) - (B^{-1}(q) - B^{-1}(0))
dC_{i_B}^-(t) \right],
\label{midPriceIncrement}
\end{align}
\begin{align}
dS(t) &= \Delta P \left[ (A^{-1}(q) - A^{-1}(0)) dM^+(t) + (B^{-1}(q) -
B^{-1}(0)) dM^-(t)  \right.\notag\\
&- \sum_{i=1}^K{(A^{-1}(0)-i)_+dL_i^+(t)} -
\sum_{i=1}^K{(B^{-1}(0)-i)_+dL_i^-(t)} \notag\\
&+ \left.  (A^{-1}(q) - A^{-1}(0)) dC_{i_A}^+(t) + (B^{-1}(q) - B^{-1}(0))
dC_{i_B}^-(t) \right].
\end{align}
The equations above are interesting in that they relate in an explicit way the
profile of the order book to the size of an increment of the mid-price or the
spread, therefore linking the price dynamics to the order flow. For instance the
infinitesimal drifts of the mid-price and the spread, conditional on the
shape of the order book at time $t$, are given by:
\begin{align}
\mathbb{E}\left[dP(t)|(\mathbf{a};\mathbf{b})\right] &= \frac{\Delta P}{2}
\left[  (A^{-1}(q) - A^{-1}(0)) \lambda^{M^+} - (B^{-1}(q) - B^{-1}(0))
\lambda^{M^-} \right.\notag\\
& - \sum_{i=1}^K{(A^{-1}(0)-i)_+ \lambda_i^{L^+}} + \sum_{i=1}^K{(B^{-1}(0)-i)_+
\lambda_i^{L^-}}\notag\\
& + \left.  (A^{-1}(q) - A^{-1}(0)) \lambda_{i_A}^{C^+}  a_{i_A} -
(B^{-1}(q) - B^{-1}(0)) \lambda_{i_B}^{C^-} |b_{i_B}| \right]dt,
\label{MidPriceDrift}
\end{align}
\noindent and
\begin{align}
\mathbb{E}\left[dS(t)|(\mathbf{a};\mathbf{b})\right] &= \Delta P \left[
(A^{-1}(q) - A^{-1}(0)) \lambda^{M^+} + (B^{-1}(q) - B^{-1}(0)) \lambda^{M^-}
\right.\notag\\
& -  \sum_{i=1}^K{(A^{-1}(0)-i)_+ \lambda_i^{L^+}} -
\sum_{i=1}^K{(A^{-1}(0)-i)_+ \lambda_i^{L^-}} \notag\\
& + \left.  (A^{-1}(q) - A^{-1}(0)) \lambda_{i_A}^{C^+}  a_{i_A} +
(B^{-1}(q) - B^{-1}(0)) \lambda_{i_B}^{C^-} |b_{i_B}| \right] dt.
\label{SpreadDrift}
\end{align}

\section{Ergodicty and Diffusive Limit}
\label{ErgodicityAndDiffusiveLimit}

In this section, our interest lies in the following questions:
\begin{enumerate}
    \item Is the order book model defined above \emph{stable}?
    \item What is the \emph{stochastic-process limit} of the price at large time
scales?
\end{enumerate}
The notions of ``stability'' and ``large-scale limit'' will be made precise
below. We first need some useful definitions from the theory of Markov chains
and stochastic stability.  Let $(Q^t)_{t \geq 0}$ be the Markov transition
probability function of the order book at time $t$, that is
\begin{equation}
     Q^t(\mathbf{x},E) := \mathbb{P}\left[ \mathbf{X}({t}) \in E |
\mathbf{X}(0)=\mathbf{x}\right], \; t \in \mathbb{R}_+, \mathbf{x}\in
\mathcal{S}, E \subset \mathcal{S},
\end{equation}
where $\mathcal{S} \subset \mathbb{Z}^{2K}$ is the state space of the order
book. We recall that a (aperiodic, irreducible) Markov process is \emph{ergodic}
if an invariant probability measure $\pi$ exists and
\begin{equation}
    \lim_{t\rightarrow\infty}|| Q^t(\mathbf{x},.)-\pi(.)||=0, \forall
\mathbf{x}\in \mathcal{S},
\end{equation}
where $||.||$ designates for a signed measure $\nu$ the \emph{total variation
norm}\footnote{The convergence in total variation norm implies the more familiar
pointwise convergence
\begin{equation}
\lim_{t \rightarrow \infty}|Q^t(\mathbf{x},\mathbf{y})-\pi(\mathbf{y})|=0,
\mathbf{x},\mathbf{y} \in \mathcal{S}.
\end{equation}
Note that since the state space $\mathcal{S}$ is countable, one can formulate
the results without the need of a ``measure-theoretic'' framework. We prefer to
use this setting as it is more flexible, and can accommodate possible
generalizations of these results.
} defined as
\begin{equation}
    ||\nu||:=\sup_{f:|f|<1}|\nu(f)| = \sup_{E \in \mathcal{B}(\mathcal{S})}
\nu(E) - \inf_{E \in \mathcal{B}(\mathcal{S})} \nu(E).
    \label{totalvariationnorm}
\end{equation}
In \eqref{totalvariationnorm}, $\mathcal{B}(\mathcal{S})$ is the Borel
$\sigma$-field generated by $\mathcal{S}$, and for a measurable function $f$ on
$\mathcal{S}$,
$\nu(f):=\int_{\mathcal{S}}fd\nu.$

\emph{$V$-uniform ergodicity.}
A Markov process is said \emph{$V$-uniformly ergodic} if there exists a
coercive\footnote{That is, a function such that $V(\mathbf{x}) \to \infty$ as
$|\mathbf{x}|\to \infty$.} function $V>1$, an invariant distribution $\pi$, and
constants $0~<~r~<~1$, and $R<\infty$ such that
\begin{equation}
||Q^t(\mathbf{x},.) - \pi(.) || \leq R r^t V(\mathbf{x}), \mathbf{x} \in
\mathcal{S}, t\in\mathbb{R}_+.
\end{equation}
$V-$uniform ergodicity can be characterized in terms of the infinitesimal
generator of the Markov process. Indeed, it is shown in
\cite{MeynTweedieBook,MeynTweedie3} that it is equivalent to the existence of a
coercive function $V$ (the ``Lyapunov test function'') such that
\begin{equation}
\mathcal{L} V(\mathbf{x}) \leq -\beta V(\mathbf{x}) + \gamma, \; \;
\text{(Geometric drift condition.)}
\label{GeometricDriftCondition}
\end{equation}
for some positive constants $\beta$ and $\gamma$. (Theorems 6.1 and 7.1 in
\cite{MeynTweedie3}.) Intuitively, condition \eqref{GeometricDriftCondition}
says that the larger $V(\mathbf{X}(t))$ the stronger $\mathbf{X}$ is pulled back
towards the center of the state space $\mathcal{S}$. A similar drift condition
is available for discrete-time Markov processes
$(\mathbf{X}_n)_{n\in\mathbb{N}}$ and reads
\begin{equation}
\mathcal{D} V(\mathbf{x}) \leq -\beta V(\mathbf{x}) + \gamma
\mathbb{I}_{\mathcal{C}}(\mathbf{x}),
\end{equation}
where $\mathcal{D}$ is the \emph{drift operator}
\begin{equation}
    \mathcal{D}V(\mathbf{x}) := \mathbb{E}[V(\mathbf{X}_{n+1}) -
V(\mathbf{X}_n) | \mathbf{X}_{n} = \mathbf{x}].
\end{equation}
and $\mathcal{C} \subset \mathcal{S}$ a finite set. (Theorem 16.0.1 in
\cite{MeynTweedieBook}.) We refer to \cite{MeynTweedieBook} for further details.

\subsection{Ergodicity of the order book and rate of convergence to the
stationary state}

Of utmost interest is the behavior of the order book in its stationary state. We have the following result:

\begin{theorem}
If $\underline{\lambda_C} = \min_{1\leq i \leq K}\{\lambda_i^{C^{\pm}}\}>0$,
then $(\mathbf{X}(t))_{t \geq 0} = (\mathbf{a}(t);\mathbf{b}(t))_{t \geq 0}$ is
an ergodic Markov process. In particular $(\mathbf{X}(t))$ has a
\emph{stationary distribution} $\pi$. Moreover, the rate of convergence of the
order book to its stationary state is \emph{exponential}. That is, there exist
$r < 1$ and $R < \infty$ such that
    \begin{equation}
        ||Q^t(\mathbf{x},.) -  \pi(.)|| \leq R r^t V(\mathbf{x}), t \in
\mathbb{R}^+, \mathbf{x} \in \mathcal{S}. \label{expo}
    \end{equation}
    \label{theo1}
\end{theorem}

\begin{proof}
Let
\begin{equation}
V(\mathbf{x}) := V(\textbf{a};\textbf{b}) :=   \sum_{i=1}^{K}{a_i} +
\sum_{i=1}^{K}{|b_i|} + q
\end{equation}
be the total number of shares in the book ($+ q$ shares).
Using the expression of the infinitesimal generator \eqref{infgen} we have
\begin{align}
\mathcal{L} V (\mathbf{x}) &\leq - (\lambda^{M^+} + \lambda^{M^-})q +
\sum_{i=1}^{K} (\lambda_i^{L^+}+\lambda_i^{L^-})q - \sum_{i=1}^{K}(
\lambda_i^{C^+} a_i + \lambda_i^{C^-} |b_i|)q\nonumber\\
&  + \sum_{i=1}^K{ \lambda_i^{L^+} (i_S-i)_+ a_{\infty}} + \sum_{i=1}^K
{\lambda_i^{L^+}(i_S-i)_+ |b_{\infty}|} \label{proof}\\
&\leq - (\lambda^{M^+} + \lambda^{M^-})q  + (\Lambda^{L^-} + \Lambda^{L^+})q -
\underline{\lambda^C} q V(\mathbf{x}) \nonumber \\
&+ K (\Lambda^{L^-} a_{\infty} + \Lambda^{L^+} |b_{\infty}| ),\label{proof2}
\end{align}
where
\begin{equation}
\Lambda^{L^{\pm}}:=\sum_{i=1}^K{\lambda_i^{L^{\pm}}}
\text{ and }
\underline{\lambda^C}:=\min_{1\leq i \leq K}{\{\lambda_i^{C^{\pm}}\}} > 0.
\end{equation}
The first three terms in the right hand side of inequality \eqref{proof}
correspond respectively to the arrival of a market, limit or cancellation
order---ignoring the effect of the shift operators. The last two terms are due
to shifts occurring after the arrival of a limit order inside the spread. The
terms due to shifts occurring after market or cancellation orders (which we do
not put in the r.h.s. of \eqref{proof}) are negative, hence the inequality. To
obtain inequality \eqref{proof2}, we used the fact that the spread $i_S$ is
bounded by $K+1$---a consequence of the boundary conditions we impose--- and
hence $(i_S-i)_+$ is bounded by $K$.

The drift condition \eqref{proof2} can be rewritten as
\begin{equation}
    \mathcal{L} V (\mathbf{ \mathbf{x} }) \leq -\beta V(\mathbf{x}) +
\gamma, \label{driftcondition}
\end{equation}
for some positive constants $\beta, \gamma$.
Inequality \eqref{driftcondition} together with theorem 7.1 in
\cite{MeynTweedie3} let us assert that $(\mathbf{X}(t))$ is $V$-uniformly
ergodic, hence \eqref{expo}.
\end{proof}

\begin{corollary}
The spread $S(t) = A^{-1}(0,\mathbf{a}(t)) \Delta P=S(\mathbf{X}(t))$ has a
well-defined stationary distribution---This is expected as by construction the
spread is bounded by $K+1$.
\end{corollary}

\subsection{The embedded Markov chain}

Let $(\mathbf{X}_n)$ denote the \emph{embedded} Markov chain associated with
$(\mathbf{X}(t))$. In event time, the probabilities of each event are
``normalized'' by the quantity
\begin{equation}
    \Lambda(\mathbf{x}) := \lambda^{M^+} + \lambda^{M^-} + \Lambda^{L^+} +
\Lambda^{L^-} +  \sum_{i=1}^K \lambda^{C^+}_i a_i + \sum_{i=1}^K \lambda^{C^-}_i
|b_i|.
\end{equation}
For instance, the probability of a buy market order when the order book is in
state $\mathbf{x}$, is
\begin{equation}
    \mathbb{P}[\text{``Buy market order at time
$n$''}|\mathbf{X}_{n-1}=\mathbf{x}]:= p^{M^+}(\mathbf{x}) =
\frac{\lambda^{M^+}}{\Lambda(\mathbf{x})}.
\end{equation}
The choice of the test function $V(\mathbf{x})=\sum_i a_i + \sum_i b_i+q$ does
not yield a geometric drift condition, and more care should be taken to obtain a
suitable test function. Let $z>1$ be a fixed real number and consider the
function\footnote{To save notations, we always use the letter $V$ for the test
function.}
\begin{equation}
    V(\mathbf{x}) := z^{\sum_i a_i + \sum_i |b_i|}:=z^{\varphi(\mathbf{x})}.
\end{equation}
We have
\begin{theorem}$(\mathbf{X}_n)$ is $V$-uniformly ergodic. Hence, there exist
$r_2 < 1$ and $R_2 < \infty$ such that
    \begin{equation}
        ||U^n(\mathbf{x},.) -  \nu(.)|| \leq R_2 r_2^n V(\mathbf{x}), n
\in \mathbb{N}, , \mathbf{x} \in \mathcal{S}.\label{expoprime}
    \end{equation}
    where $(U^n)_{n \in \mathbb{N}}$ is the transition probability function
of $(\mathbf{X}_n)_{n \in \mathbb{N}}$ and $\nu$ its stationary distribution.
\end{theorem}
\begin{proof}
\begin{eqnarray}
\mathcal{D}V(\mathbf{x}) & \leq & \frac{\lambda^{M^+}}{\Lambda(\mathbf{x})} (
z^{\sum_i a_i - q + \sum_i |b_i|} - V(\mathbf{x}) ) \nonumber \\
 & + & \sum_j \frac{\lambda^{L^+}_j}{\Lambda(\mathbf{x})} ( z^{\sum_i a_i + q +
\sum_i |b_i| + K |b_{\infty}|} - V(\mathbf{x}) ) \nonumber \\
 & + & \sum_j \frac{\lambda^{C^+}_j a_j}{\Lambda(\mathbf{x})} (z^{\sum_i a_i - q
+ \sum_i |b_i|} - V(\mathbf{x})) \nonumber \\
 & + & \frac{\lambda^{M^-}}{\Lambda(\mathbf{x})} (z^{\sum_i a_i + \sum_i |b_i| -
q} - V(\mathbf{x})) \nonumber \\
 & + & \sum_j \frac{\lambda^{L^-}_j}{\Lambda(\mathbf{x})} (z^{\sum_i a_i + K
a_{\infty} + \sum_i |b_i| + q} - V(\mathbf{x})) \nonumber \\
 & + & \sum_j \frac{\lambda^{C^-}_j |b_j|}{\Lambda(\mathbf{x})} (z^{\sum_i a_i +
\sum_i |b_i| - q} - V(\mathbf{x})).
 \label{Majoration}
\end{eqnarray}
If we factor out $V(\mathbf{x})=z^{\sum a_i + \sum b_i}$ in the r.h.s of
\eqref{Majoration}, we get
\begin{eqnarray}
\frac{\mathcal{D}V(\mathbf{x})}{V(\mathbf{x})} & \leq & \frac{\lambda^{M^+} +
\lambda^{M^-}}{\Lambda(\mathbf{x})} (z^{-q}-1) \nonumber \\
&+& \frac{\Lambda^{L^-} + \Lambda^{L^-}}{\Lambda(\mathbf{x})} (z^{q+K
d_{\infty}}-1) \nonumber \\
& + & \frac{\sum_j \lambda^{C^+}_j  a_j + \sum_j \lambda^{C^-}_j
|b_j|}{\Lambda(\mathbf{x})} (z^{-q} - 1),
\label{Drift}
\end{eqnarray}
where
\begin{equation}
    d_{\infty} := \max\{a_{\infty}, |b_{\infty}|\}.
\end{equation}
Then
\begin{eqnarray}
 \frac{\mathcal{D}V(\mathbf{x})}{V(\mathbf{x})}  & \leq & \frac{\lambda^{M^+} +
\lambda^{M^-}}{\lambda^{M^+} + \lambda^{M^-} + \Lambda^{L^+} + \Lambda^{L^-}  +
\overline{\lambda^C} \varphi(\mathbf{x})}  (z^{-q}-1) \nonumber \\
&+& \frac{\Lambda^{L^+} + \Lambda^{L^-}}{\lambda^{M^+} + \lambda^{M^-} +
\Lambda^{L^+} + \Lambda^{L^-} + \underline{\lambda^C} \varphi(\mathbf{x})}
(z^{q+K d_{\infty}}-1) \nonumber \\
& + & \frac{\underline{\lambda^C} \varphi(\mathbf{x}) }{\lambda^{M^+} +
\lambda^{M^-} + \Lambda^{L^+} + \Lambda^{L^-}  +  \overline{\lambda^C}
\varphi(\mathbf{x})} (z^{-q}-1),
\label{EventTimeDrift}
\end{eqnarray}
with the usual notations
\begin{equation}
\underline{\lambda^C} := \min{\lambda_i^{C^{\pm}}} \text{ and }
\overline{\lambda^C} := \max{\lambda_i^{C^{\pm}}}.
\end{equation}
Denote the r.h.s of \eqref{EventTimeDrift} $B(\mathbf{x})$. Clearly
\begin{equation}
    \lim_{\varphi(\mathbf{x}) \rightarrow \infty} B(\mathbf{x}) =
\frac{\underline{\lambda^C} (z^{-q} -1)} {\overline{\lambda^C}} < 0,
\end{equation}
hence there exists $A>0$ such that for $\mathbf{x} \in \mathcal{S}$ and
$\varphi(\mathbf{x}) > A$
\begin{equation}
    \frac{\mathcal{D}V(\mathbf{x})}{V(\mathbf{x})} \leq
\frac{\underline{\lambda^C} (z^{-q} -1)} {2 \overline{\lambda^C}} := - \beta <
0.
\end{equation}
Let $\mathcal{C}$ denote the finite set
\begin{equation}
\mathcal{C} = \{\mathbf{x} \in \mathcal{S} : \varphi(\mathbf{x}) = \sum_i a_i +
\sum_i b_i \leq A \}.
\end{equation}
We have
\begin{equation}
\mathcal{D} V(\mathbf{x}) \leq - \beta V(\mathbf{x}) + \gamma
\mathbb{I}_{\mathcal{C}}(\mathbf{x}),
\end{equation}
with
\begin{equation}
\gamma := \max_{\mathbf{\mathbf{x}}\in
\mathcal{C}}{\mathcal{D}V(\mathbf{\mathbf{x}})}.
\end{equation}
Therefore $(\mathbf{X}_n)_{n\geq0}$ is $V$-uniformly ergodic, by theorem 16.0.1
in \cite{MeynTweedieBook}.
\end{proof}

\subsection{The case of constant cancellation rates}
The proof above can be applied to the case where the cancellation rates do not depend on
 the state of the order book $\mathbf{X'}(t)$---We shall denote the
order book $\mathbf{X}'(t)$ in order to highlight that the assumption of
proportional cancellation rates is relaxed. The probability of a cancellation
$dC^{\pm}_i(t)$ in $[t, t + \delta t]$ is now
\begin{equation}
\mathbb{P}[C^{\pm}_i(t + \delta t) - C^{\pm}_i(t) =
1|\mathbf{X'}(t)=\mathbf{x'}] = \lambda^{C^{\pm}}_i \delta t + o(\delta t),
\end{equation}
instead of
\begin{equation}
\mathbb{P}[C^{+}_i(t + \delta t) - C^{+}_i(t) =
1|\mathbf{X'}(t)=\mathbf{x'}] = \lambda^{C^{+}}_i a_i(t) \delta t + o(\delta t),
\end{equation}
where $\lim_{\delta t \to 0} {o(\delta t)}/{\delta t} =0.$
Since $\Lambda = \lambda^{M^+} + \lambda^{M^-} + \Lambda^{L^+} + \Lambda^{L^-} +
 \sum_{i=1}^K \lambda^{C^+}_i + \sum_{i=1}^K \lambda^{C^-}_i$ does not depend on
$\mathbf{x'}$, the analysis of the stability of the continuous-time process
$(\mathbf{X}'(t))$ and its discrete-time counterpart $(\mathbf{X}'_n)$ are
essentially the same.

We have the following result:
\begin{theorem}
Set
\begin{equation}
\Lambda^{C^{\pm}}:=\sum_{i=1}^K \lambda_i^{C^{\pm}} \text{ and }
\Lambda^{L^{\pm}}:=\sum_{i=1}^K \lambda_i^{L^{\pm}}.
\end{equation}
Under the condition
\begin{equation}
\lambda^{M^+} + \lambda^{M^-}   + \Lambda^{C^+} + \Lambda^{C^-} > (\Lambda^{L^+}
+ \Lambda^{L^-}) (1+K d_{\infty}),
\label{condition}
\end{equation}
$(\mathbf{X}'_n)$ is $V$-uniformly ergodic. There exist $r_3 < 1$ and $R_3 <
\infty$ such that
    \begin{equation}
        ||U'^n(\mathbf{x},.) - {\nu}'(.)|| \leq R_3 r_3^n V(\mathbf{x}), n \in
\mathbb{N}, \mathbf{x} \in \mathcal{S}.
    \end{equation}
The same is true for $(\mathbf{X}'(t))$.
\end{theorem}

\begin{proof}
Let us prove the result for $(\mathbf{X}'_n)$. Inequality \eqref{Drift} is still
valid by the same arguments, but this time the arrival rates are independent of
$\mathbf{x'}$
\begin{eqnarray}
 \frac{\mathcal{D}V(\mathbf{x'})}{V(\mathbf{x'})}  & \leq & \frac{\lambda^{M^+}
+ \lambda^{M^-}}{\Lambda}  (z^{-q}-1) \nonumber \\
&+& \frac{\Lambda^{L^+} + \Lambda^{L^-}}{\Lambda} (z^{q+K d_{\infty}}-1)
\nonumber \\
& + & \frac{\Lambda^{C^+} + \Lambda^{C^-}} {\Lambda} (z^{-q}-1).
\end{eqnarray}
Set
\begin{equation}
z =: 1 + \epsilon > 1.
\end{equation}
A Taylor expansion in $\epsilon$ gives
\begin{eqnarray}
 \Lambda \frac{\mathcal{D}V(\mathbf{x})}{V(\mathbf{x})}  & \leq & (\lambda^{M^+}
+ \lambda^{M^-}) (-q \epsilon) \nonumber \\
 & + &  (\Lambda^{L^+} + \Lambda^{L^-}) (q+K d_{\infty}) \epsilon \nonumber \\
 & + & (\Lambda^{C^+} + \Lambda^{C^-}) (- q \epsilon) + o(\epsilon).
 \label{DL}
\end{eqnarray}
For $\epsilon>0$ small enough, the sign of \eqref{DL} is determined by the
quantity
\begin{equation}
-(\lambda^{M^+} + \lambda^{M^-})  +  (\Lambda^{L^+} + \Lambda^{L^-}) (1+K
d_{\infty}) - (\Lambda^{C^+} + \Lambda^{C^-}).
\end{equation}
Hence, if \eqref{condition} holds
\begin{eqnarray}
 \mathcal{D}V(\mathbf{x}) & \leq - \beta V(\mathbf{x}) \;\text{ for some } \beta
>0,
\end{eqnarray}
and a geometric drift condition is obtained for $\mathbf{X}'$.
\end{proof}
If for concreteness we set $q = 1$ share, and all the arrival rates are symmetric and do not depend on $i$, then condition \eqref{condition} can be rewritten as
\begin{equation}
\lambda^M + K \lambda^C > K \lambda^L (1 + K d_{\infty}).
\end{equation}
where $K$ is the size of the order book and $d_{\infty}$ is the depth far away
from the mid-price. Note that the above is a \emph{sufficient} condition for
($V$-uniform) stability.

\subsection{Large-scale limit of the price process}

We are now able to answer the main question of this paper. Let us define the
process $$e(t) \in \{1, \dots, 2 (2 K+1)\}$$ which indicates the last event
\begin{equation}
\{M^{\pm}, L_i^{\pm},C_i^{\pm}\}_{i\in\{1,\dots,K\}},\nonumber
\end{equation}
that has occurred before time $t$.
\begin{lemma}
If we append $e(t)$ to the order book $(\mathbf{X}(t))$, we get a Markov
process
\begin{equation}
\mathbf{Y}(t):=(\mathbf{X}(t),e(t))
\end{equation}
which still satisfies the drift condition \eqref{GeometricDriftCondition}.
\end{lemma}
\begin{proof}
Since $e(t)$ takes its values in a finite set, the arguments of the previous
sections are valid with minor modifications, and with the test functions
\begin{eqnarray}
V(\mathbf{y}) &:=& q+ \sum a_i + \sum |b_i| + e, \; \text{ (continuous-time
setting)}\\
V(\mathbf{y}) &:=& e^{\sum a_i + \sum |b_i| + e}. \; \text{ (discrete-time
setting)}
\end{eqnarray}
The $V$-uniform ergodicity of $(\mathbf{Y}(t))$ and $(\mathbf{Y}_n)$ follows.

\end{proof}
Given the state $\mathbf{X}_{n-1}$ of the order book at time $n-1$ and the event
$e_n$, the price increment at time $n$ can be determined. (See equation
\eqref{midPriceIncrement}.)  We define the sequence of random variables
\begin{equation}
\eta_n := \Psi(\mathbf{X}_{n-1}, e_n) := \Phi(\mathbf{Y}_n,\mathbf{Y}_{n-1}),
\end{equation}
as the price increment at time $n$. $\Psi$ is a deterministic function giving
the elementary ``price-impact'' of event $e_n$ on the order book at state
$\textbf{X}_{n-1}$. Let $\mu$ be the stationary distribution of
$(\mathbf{Y}_n)$, and $M$ its transition probability function. We are interested
in the random sums
\begin{equation}
    P_n :=  \sum_{k=1}^{n} \overline{\eta}_n = \sum_{k=1}^{n}
\overline{\Phi}(\mathbf{Y}_k,\mathbf{Y}_{k-1}),
\end{equation}
where
\begin{equation}
\overline{\eta}_k := \eta_k - \mathbb{E}_{\mu}[\eta_k]  = \overline{\Phi}_k =
\Phi_k - \mathbb{E}_{\mu}[\Phi_k],
\end{equation}
and the asymptotic behavior of the rescaled-centered price process
\begin{equation}
    \widetilde{P}^{(n)}(t) := \frac{P_{\lfloor n t \rfloor}}{\sqrt n },
\end{equation}
as $n$ goes to infinity.
\begin{theorem}
The series
\begin{equation}
    \sigma^2 = \mathbb{E}_{\mu}[\overline{\eta}_0^2] + 2
\sum_{n=1}^{\infty}\mathbb{E}_{\mu}[\overline{\eta}_0 \overline{\eta}_n]
\end{equation}
converges absolutely, and
the rescaled-centered price process is a Brownian motion in the limit of $n$
going to infinity. That is
    \begin{equation}
        \widetilde{P}^{(n)}(t) \stackrel{n\to\infty}{\longrightarrow}
\sigma B(t),
    \end{equation}
where $(B(t))$ is a standard Brownian motion.
\label{MainResult}
\end{theorem}

\begin{proof}
The idea is to apply the functional central limit theorem for (stationary and
ergodic) sequences of \emph{weakly dependent} random variables with \emph{finite
variance}. Firstly, we note that the variance of the price increments $\eta_n$ is
finite since it is bounded by $K+1$. Secondly, the $V$-uniform ergodicity of
$(\mathbf{Y}_n)$ is equivalent to
\begin{equation}
        ||M^n(\mathbf{x},.) - \mu(.)|| \leq R \rho^n V(\mathbf{x}), n
\in \mathbb{N},
\end{equation}
for some $R<\infty$ and $\rho<1.$
 This implies thanks to theorem 16.1.5 in \cite{MeynTweedieBook} that for any
$g^2, h^2 \leq V$, $k,n \in \mathbb{N}$, and any initial condition $\mathbf{y}$
\begin{equation}
    |\mathbb{E}_{\mathbf{y}}[ g(\mathbf{Y}_k) h(\mathbf{Y}_{n+k})] -
\mathbb{E}_{\mathbf{y}}[g(\mathbf{Y}_k)]
\mathbb{E}_{\mathbf{y}}[g(\mathbf{Y}_k)]| \leq R \rho^n [1+\rho^k
V(\mathbf{y})],
\end{equation}
where $\mathbb{E}_{\mathbf{y}}[.]$ means
$\mathbb{E}[.|\mathbf{Y}_0=\mathbf{y}]$.
This in turn implies
\begin{equation}
    | \mathbb{E}_{\mathbf{y}}[\overline{h}(\mathbf{Y}_k)
\overline{g}(\mathbf{Y}_{k+n})] | \leq R_1 \rho^n [1+\rho^k V(\mathbf{y})]
    \label{geometricx}
\end{equation}
for some $R_1 < \infty$, where $\overline{h} = h - \mathbb{E}_{\mu}[h],
\overline{g} = g - \mathbb{E}_{\mu}[g]$.
By taking the expectation over $\mu$ on both sides of \eqref{geometricx} and
noting that $\mathbb{E}_{\mu}[V(\mathbf{Y}_0)]$ is finite by theorem 14.3.7 in
\cite{MeynTweedieBook}, we get
\begin{equation}
    |\mathbb{E}_{\mu}[\overline{h}(\mathbf{Y}_k)
\overline{g}(\mathbf{Y}_{k+n})]| \leq R_2 \rho^n =: \rho(n), k,n \in \mathbb{N}.
\end{equation}
Hence the stationary version of $(\mathbf{Y}_n)$ satisfies a \emph{geometric
mixing condition}, and in particular
\begin{equation}
\sum_{n}\rho(n)<\infty.
\label{integrable}
\end{equation}
Theorems 19.2 and 19.3 in \cite{Billingsley2} on functions of mixing
processes\footnote{See also theorem 4.4.1 in \cite{Whitt} and discussion
therein.} let us conclude that
\begin{equation}
    \sigma^2 := \mathbb{E}_{\mu}[\overline{\eta}_0^2] + 2
\sum_{n=1}^{\infty}\mathbb{E}_{\mu}[\overline{\eta}_0 \overline{\eta}_n]
\label{asymptoticvariance}
\end{equation}
is well-defined---the series in \eqref{asymptoticvariance} converges
absolutely---and coincides with the asymptotic variance
\begin{equation}
        \lim_{n \rightarrow \infty}\frac{1}{n}\mathbb{E}_{\mu}\left[
\sum_{k=1}^n{(\overline{\eta}_k)^2} \right] = \sigma^2.
\end{equation}
Moreover
\begin{equation}
        \widetilde{P}^{(n)}(t) \stackrel{n\to\infty}{\longrightarrow}
\sigma B(t), \label{Brownian}
\end{equation}
where $(B(t))$ is a standard Brownian motion. The convergence in
\eqref{Brownian} happens in $D[0, \infty)$, the space of $\mathbb{R}$-valued
c\`adl\`ag functions, equipped with the Skorohod topology.
\end{proof}

\begin{remark}
In the large-scale limit, the mid-price $P$, the ask price $P^A = P +
\frac{S}{2}$, and the bid price $P^B = P - \frac{S}{2}$ converge to the same
process $(\sigma B(t))$.
\end{remark}

\begin{remark}
Theorem \ref{MainResult} is also true with constant
cancellation rates under condition \eqref{condition}. In this case the result
holds both in event time and physical time. Indeed, let $(N(t))_{t \in
\mathbb{R}_+}$ denote a Poisson process with intensity $\Lambda =
\lambda^{M^{\pm}} + \Lambda^{L^{\pm}} +  \sum_{i=1}^K \lambda^{C^{\pm}}_i$. The
price process in physical time $(P_c(t))_{t \in \mathbb{R}_+}$ can be linked to
the price in event time $(P_n)_{n \in \mathbb{N}}$ by
\begin{equation}
P_c(t)=P_{N(t)}.
\end{equation}
Then
\begin{eqnarray}
\frac{P_{\lfloor k t \rfloor}}{\sqrt{k}} & \stackrel{k\to\infty}
{\longrightarrow} \sigma B(t)\; \text{ as in theorem \ref{MainResult}, } \\
&\text{and since } \frac{N(u)}{\Lambda u} \stackrel{u\to\infty}
{\longrightarrow} 1 \text{ a.s., } \notag \\
\frac{P_c( k t)}{\sqrt{k}}   & = \frac{P_{N(kt)}}{\sqrt{k}}
\stackrel{k\to\infty}{\sim} {\rightarrow}  \frac{P_{\lfloor \Lambda kt
\rfloor}}{\sqrt{k}}   \stackrel{k\to\infty}{\longrightarrow}  \sqrt{\Lambda}
\sigma B(t).
\end{eqnarray}
\end{remark}

\begin{remark}
\emph{Yet another specification of the cancellation process.} Another interesting specification of the cancellation process $(C_i(t))$ is to assume that the arrival rate is constant (for each $i$) but the canceled \emph{volume} is proportional to the queue size $|\mathbf{X}_i|$.  In this case, the treatments of the continuous time chain and its embedded discrete-time counterpart are equivalent, and theorems \ref{theo1}--\ref{MainResult} can be obtained in an analogous manner to the proofs in this section.
\end{remark}

\section{Numerical Example}
\label{NumericalExample}

In order to gain a better intuitive understanding of the ``mechanics'' of the model, we sketch in Algorithm \ref{algo1} below the simulation procedure in pseudo-code (See also \cite{GatheralOomen} for a similar description.) For simplicity, we take a symmetric order book model. We also let (usual notations):
\begin{eqnarray}
\boldsymbol\lambda^L &:=& \left(\lambda^{L}_1, \dots, \lambda^{L}_K \right),\\
 \Lambda^L &:=& \sum_{i=1}^K \lambda^{L}_i, \\
\boldsymbol\lambda^C(\mathbf{a}) &:=& \left(\lambda^{C}_1 a_1, \dots, \lambda^{C}_K a_K\right),\\
\Lambda^C(\mathbf{a}) &:=& \sum_{i=1}^K \lambda_i^{C} a_i,\\
\boldsymbol\lambda^C(\mathbf{b}) &:=& \left(\lambda^{C}_1 |b_1|, \dots, \lambda^{C}_K |b_K|\right),\\
\Lambda^C(\mathbf{b}) &:=& \sum_{i=1}^K \lambda_i^{C} |b_i|,\\
\Lambda(\mathbf{a}, \mathbf{b}) &:=& 2 (\lambda^M + \Lambda^L) + \Lambda^C(\mathbf{a}) + \Lambda^C(\mathbf{b}).
\end{eqnarray}
In order to put the simulation results and the data on the same footing, we relax the assumption of constant order sizes; we draw the order volumes from lognormal distributions.
\begin{algorithm}
    \begin{algorithmic}[1]
    \REQUIRE \emph{Model~parameters---} 
    Arrival rates: $\lambda^M, \{\lambda^L_i\}_{i \in \{1, \dots K\}}, \{\lambda^C_i\}_{i \in \{1, \dots K\}}$, order book size: $K$, reservoirs: $a_{\infty}$, $b_{\infty}$, volume distribution parameters: $(v^M, s^M), (v^L, s^L), (v^C, s^C)$. \\
    	\emph{Simulation Parameters---} Number of time steps: $N$.\\
    	\emph{Initialization---} $t \leftarrow 0$, $\mathbf{X}(0) \leftarrow \mathbf{X}_{\text{init}}$.
        \FOR {time step $n = {1, \dots, N}$,}
            \STATE Compute the best bid $p^B$ and best ask $p^A$.
            \STATE Compute $\displaystyle \Lambda^{C}(\mathbf{b})=\sum_{i=1}^K \lambda_i^{C} |b_i|,$ i.e. the weighted sum of shares at price levels from $p^A - K$ to $p^A - 1$.
            \STATE Compute $\displaystyle \Lambda^{C}(\mathbf{a})=\sum_{i=1}^K \lambda_i^{C} a_i.$
            \STATE Draw the waiting time $\tau$ for the next event from an exponential distribution with parameter $$\Lambda(\mathbf{a}, \mathbf{b}) = 2 (\lambda^M + \Lambda^L) + \Lambda^C(\mathbf{a}) + \Lambda^C(\mathbf{b}).$$
            \STATE Draw a new event according to the probability vector
            $$\left(\lambda^M, \lambda^M, \Lambda^L, \Lambda^L, \Lambda^C(\mathbf{a}), \Lambda^C(\mathbf{b}) \right) / \Lambda(\mathbf{a}, \mathbf{b}).$$
    These probabilities correspond respectively to a buy market order, a sell market order, a buy limit order, a sell limit order, a cancellation of an existing sell order and a cancellation of an existing buy order.
            \STATE Depending on the event type, draw the order volume from a lognormal distribution with parameters $(v^M, s^M), (v^L, s^L)$ or $(v^C, s^C)$.
            \STATE If the selected event is a limit order, select the relative price level from $\{1, 2, \dots, K\}$ according to the probability vector
            $$ \left(\lambda^{L}_1, \dots, \lambda^{L}_K \right) / \Lambda^L.$$
            \STATE If the selected event is a cancellation, select the relative price level at which to cancel an order from $\{1, 2, \dots, K\}$ according to the probability vector
            $$\left(\lambda^{C}_1 a_1, \dots, \lambda^{C}_K a_K\right) / \Lambda^C(\mathbf{a}).$$

            {(or $\boldsymbol\lambda^C(\mathbf{b})/\Lambda^C(\mathbf{b})$ for the bid side.)}
            \STATE Update the order book state according to the selected event.
            \STATE Enforce the boundary conditions:
            \begin{eqnarray}
            a_i &=& a_{\infty}, i\geq K+1, \nonumber\\
            b_i &=& b_{\infty}, i \geq K+1 \nonumber.
            \end{eqnarray}
            \STATE Increment the event time $n$ by 1 and the physical time $t$ by $\tau$.
        \ENDFOR
        \remark{For the practical implementation, it is easier to work with an ``absolute'' price frame $\Delta p \times \{1 \dots L\}$ where $L \gg K$.}
    \end{algorithmic}
    \caption{Order book simulation.}
    \label{algo1}
\end{algorithm}
The parameters of the model are estimated from tick by tick data as detailed in \ref{ModelParameters}. For concreteness\footnote{The results are qualitatively the same for all CAC 40 stocks.}, we use the parameters of the stock SCHN.PA (Schneider Electric) in March 2011 for the plots. They are summarized in tables \ref{ModelParams1} and \ref{ModelParams2}.

Figure \ref{AverageDepthProfile} represents the average depth profile, that is, the average number of outstanding shares at a distance of $i$ ticks from the best opposite price. The agreement between the simulation and the data is fairly good (See panel $(a)$ of figure \ref{Cac40Panel} for a cross-sectional view on CAC 40 stocks.) We also plot the distribution of the spread in figure \ref{SpreadDistribution}. Note that the simulated distribution is \emph{tighter} than the actual one (this is systematic and is documented in panel $(b)$ of figure \ref{Cac40Panel}.) Figure \ref{Autocorr} shows the fast decay of the autocorrelation function of the price increments. Note the high negative autocorrelation of simulated trade prices relatively to the data. In accordance with the theoretical analysis, figures \ref{PricePath}--\ref{HistPriceInc} illustrate the asymptotic normality of price increments.

The \emph{signature plot} of the price time series is defined as the variance of price increments at lag $h$ normalized by the lag, that is
\begin{equation}
    \sigma_h^2 := \frac{\mathbb{V}\left[ P(t+h) - P(t) \right]}{h}.
\end{equation}
This function measures the variance of price increments per time unit. It is interesting in that it shows the transition from the variance at small time scales where micro-structure effects dominate, to the long-term variance.
By theorem \ref{MainResult}\footnote{Strictly spreaking, we proved the result in event-time.}
\begin{equation}
\lim_{h\rightarrow\infty} \sigma_h^2 = \sigma^2, \text{ for some fixed value $\sigma.$}
\end{equation}
We verify this numerically in figure \ref{SignaturePlot}. Two remarks are in order regarding the signature plot:

\smallskip \emph{Long-term variance---} The simulated long-term variance is systematically lower than the variance computed from the data (This is documented in panel $(c)$ of figure \ref{Cac40Panel}.) Intuitively, the depth of the order book is expected to increase from the best price towards the center of the book. In the absence of autocorrelation in trade signs, this would cause prices to wander less often far away from the current best as they hit a higher ``resistance''. We also suspect that actual prices exhibit locally more ``drifting phases'' than in a (symmetric) Markovian model where the expected price drift is null at all times. An interesting analysis of a simple order book model that allows time-varying arrival rates can be found in \cite{ChalletStinchcombeQuantFi}.

 \smallskip \emph{Short-term variance---} The signature plot predicted by the model is too high at short time scales relative to the asymptotic variance, especially for traded prices. This is classically known to be due to bid-ask bounce. It is however remarkable that the signature plot of actual trade prices looks much flatter compared to the simulation (See figure \ref{SignaturePlot}.) This was discovered and discussed in detail by Bouchaud et al. in \cite{BouchaudFlucAndResp}, and Lillo and Farmer in \cite{LilloFarmer} (See also \cite{FarmerGerig} and \cite{BouchaudFarmerLillo}.) They note that actual order signs exhibit positive long-ranged correlations. They also note that actual prices are diffusive---the signature plot is flat---even at small time scales. They solve this apparent paradox by showing that diffusivity results from two opposite effects: autocorrelation in trade signs induces persistence in the prices, just at the exact amount to counterbalance the mean reversion induced by the liquidity stored in the order book. 
This subtle equilibrium between liquidity takers and liquidity providers
which guarantees price diffusivity at short lags,
is not accounted for by the bare Markovian order book model we study, and one can speak about anomalous diffusion at short time scales for Markovian order book models \cite{SmithFarmer}. Because of the absence of positive autocorrelation in trade signs in the model, this effect is magnified when one looks at trades. The next paragraph elaborates on this point.

\begin{figure}
\begin{center}
\includegraphics[width=\textwidth]{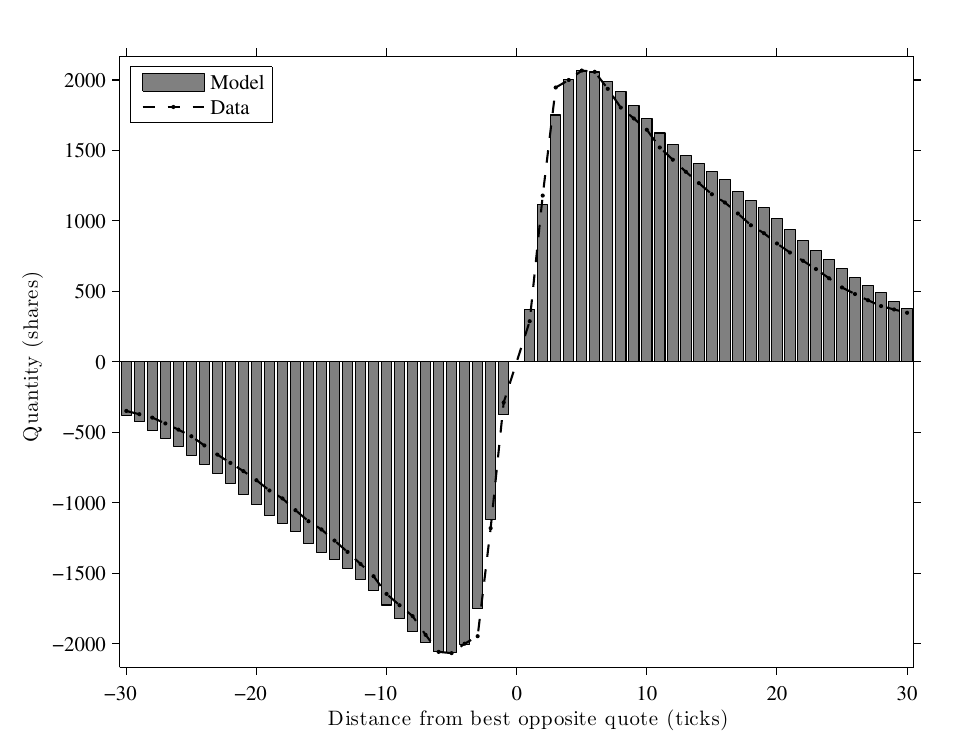}
\caption{\label{AverageDepthProfile}Average depth profile. Simulation parameters are summarized in tables \ref{ModelParams1} and \ref{ModelParams2}.}
\end{center}
\end{figure}

\begin{figure}
\begin{center}
\includegraphics[width=\textwidth]{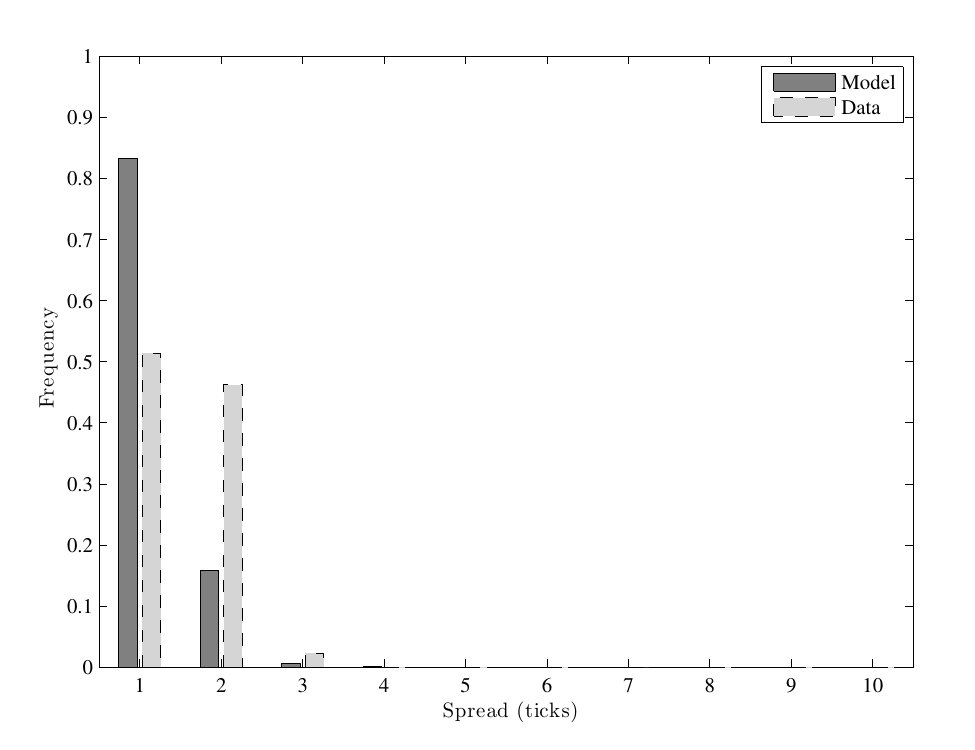}
\end{center}
\caption{\label{SpreadDistribution}Probability distribution of the spread. Note that the model (dark gray) predicts a tighter spread than the data.}
\end{figure}

\begin{figure}
\begin{center}
\includegraphics[width=\textwidth]{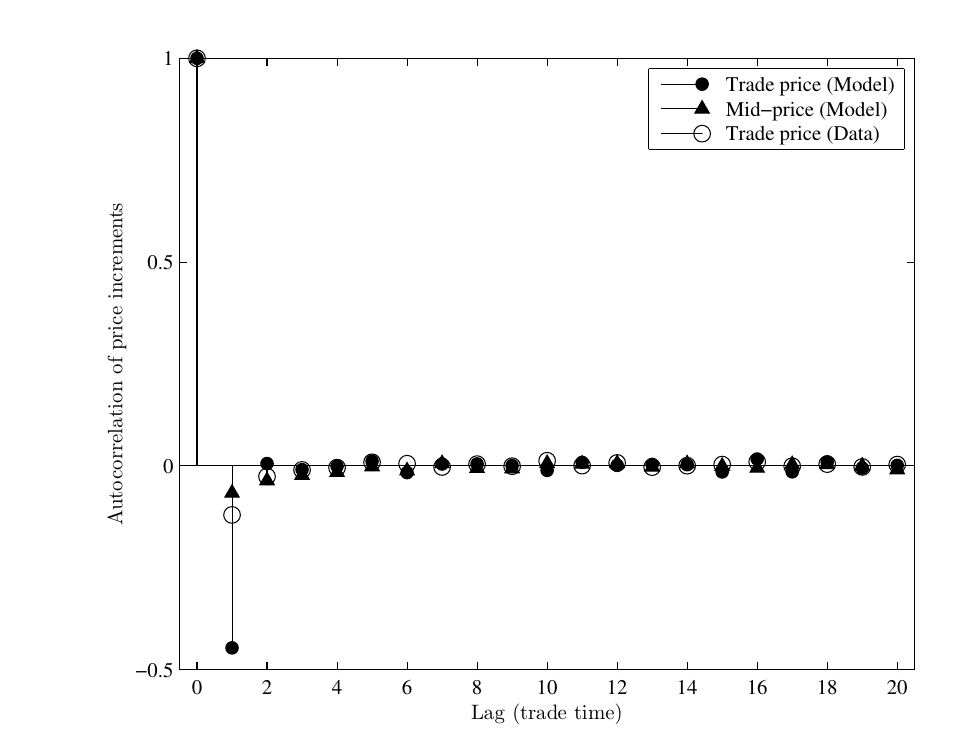}
\caption{\label{Autocorr}Autocorrelation of price increments. This figure shows the fast decay of the autocorrelation function, and the large negative autocorrelation of trades at the first lag.}
\end{center}
\end{figure}

\begin{figure}
\begin{center}
\includegraphics[width=\textwidth]{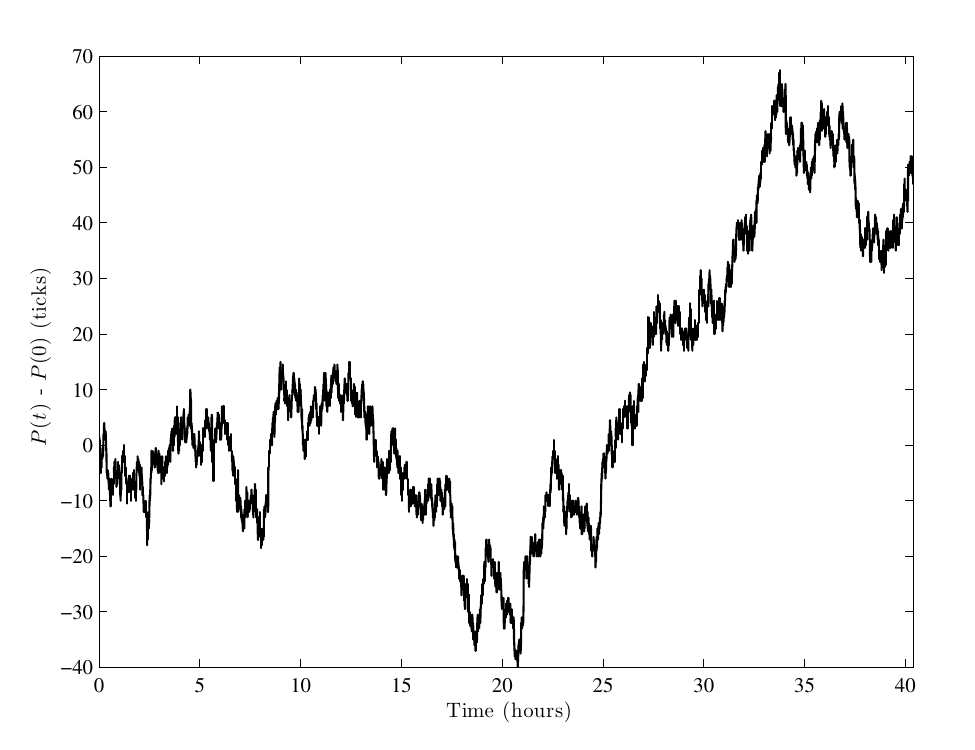}
\end{center}
\caption{\label{PricePath}Price sample path. At large time scales, the price process is a Brownian motion.}
\end{figure}

\begin{figure}
\begin{center}
\includegraphics[width=\textwidth]{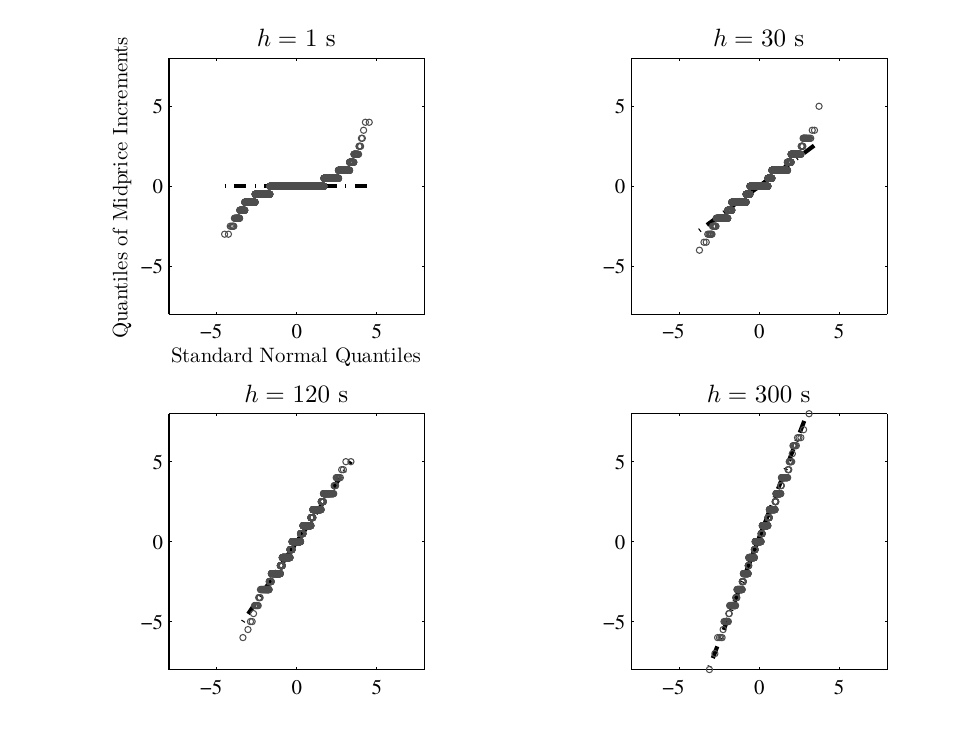}
\caption{\label{QqPlot}Q-Q plot of mid-price increments. $h$ is the time lag in seconds. This figure illustrates the aggregational normality of price increments.}
\end{center}
\end{figure}

\begin{figure}
\begin{center}
\includegraphics[width=\textwidth]{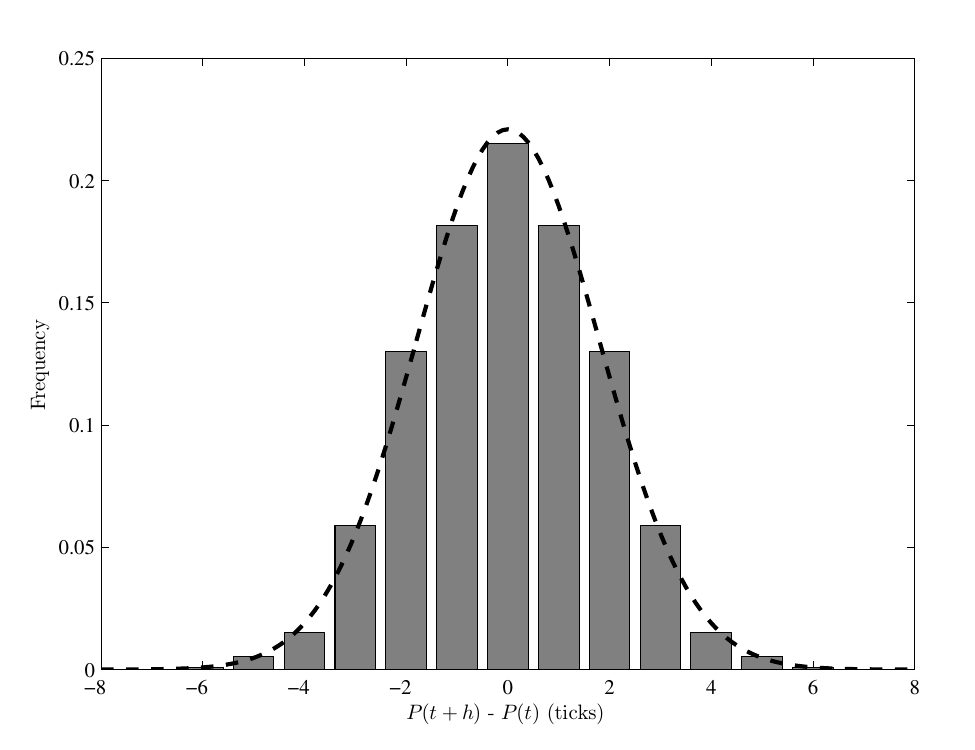}
\end{center}
\caption{\label{HistPriceInc}Probability distribution of price increments.  Time lag $h = 1000$ events.}
\end{figure}

\begin{figure}
    \centering
        \subfloat[Trade time.]{
    \label{fig:image_2}
    \includegraphics[width=0.85\textwidth]{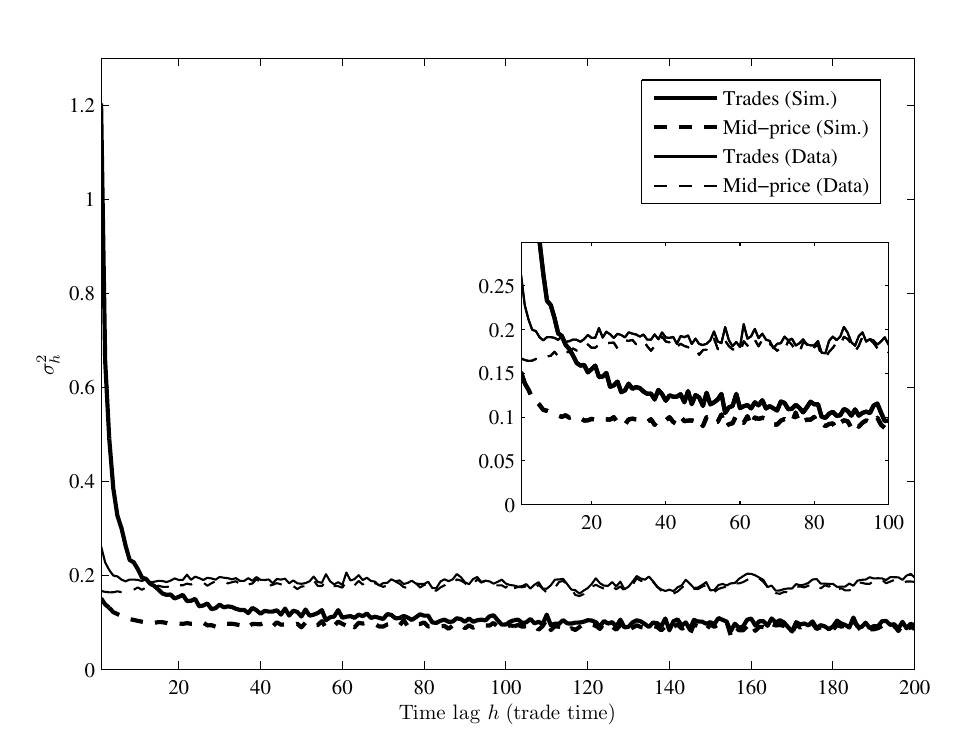}}
    
    \subfloat[Calendar time.]{
    \label{fig:image_1}
    \includegraphics[width=0.85\textwidth]{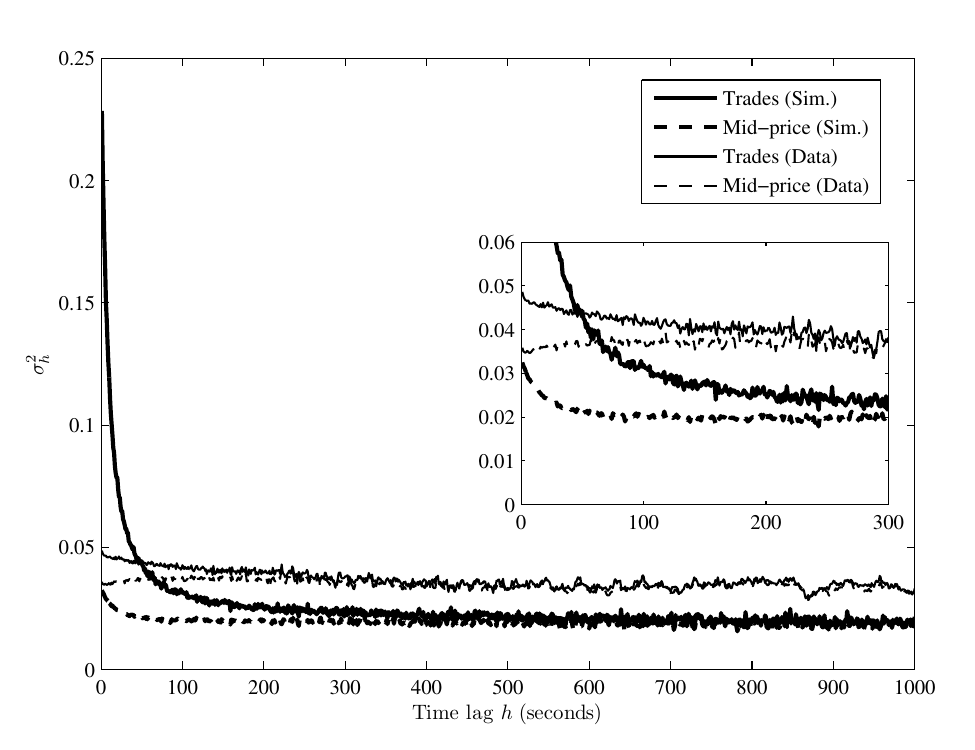}}    
    \centering
    \caption{\label{SignaturePlot}Signature plot: $\displaystyle \sigma_h^2 := {\mathbb{V}\left[ P(t+h) - P(t) \right]}/{h}.$ $y$ axis unit is $\text{tick}^2 \text{ per trade}$ for panel $(a)$ and $\text{tick}^2 . \text{second}^{-1}$ for panel $(b).$ We used a 1,000,000 event simulation run for the model signature plots. Data signature plots are computed separately for each trading day $[9:30$--$14:00]$ then averaged across 23 days. For calendar time signature plots, prices are sampled every second using the last tick rule. The inset is a zoom-in.}
\end{figure}

\smallskip \noindent \textbf{Anomalous diffusion at short time scales.} A qualitative understanding of the discrepancy between the model and the data signature plots at short time scales can be gleaned with the following heuristic argument. In what follows, we reason in trade time $t$. Denote by $P^{Tr}(t)$ the price of the trade at time $t$, and $\alpha(t)$ its sign:
\begin{equation}
\alpha(t)  =  1, \;   \text{for a buyer initiated trade, i.e. a buy market order,}
\end{equation}
and,
\begin{equation}
\alpha(t)  =   -1, \;  \text{for a seller initiated trade, i.e. a sell market order.}
\end{equation}
We assume that the two signs are equally probable (symmetric model). But to make the argument valid for both the model (for which successive trade signs are independent) and the data (for which trade signs exhibit long memory) we do not assume independence of successive trade signs. Let also for a quantity $Z$
\begin{equation}
    \Delta Z(t):= Z(t+1) - Z(t).
\end{equation}
We have by definition
\begin{equation}
    P^{Tr}(t) = P(t^-) + \frac{1}{2} \alpha(t) S(t^-),
    \label{Trade2Mid}
\end{equation}
where $P(t^-)$ and $S(t^-)$ are respectively the prevailing mid-price and spread just before the trade. From equation \eqref{Trade2Mid}
\begin{eqnarray}
{\sigma_1^{Tr}}^2 &:= & \mathbb{V}[\Delta P^{Tr}(t)] \nonumber \\
& = & \mathbb{E}\left[ \left( \Delta P^{Tr}(t) \right) ^2 \right] \nonumber  \\
& = & \mathbb{E} \left[ \left( \Delta P(t^-) \right) ^2  \right] \nonumber \\
& + & \mathbb{E} \left[ \Delta P(t^-) \Delta (\alpha(t) S(t^-)) \right] \nonumber \\
& +& \frac{1}{4} \mathbb{E} \left[ \left( \Delta (\alpha(t) S(t^-)) \right)^2 \right].
\end{eqnarray}
The first term in the r.h.s. is the variance of mid-price increments $\sigma_1^2$.
The second term represents the covariance of mid-price increments and the trade sign (weighted by the spread) and we assume it is negligible\footnote{This amounts to neglecting the correlation between trade signs and mid-quote movements, which is justified by the dominance of cancellations and limit orders in comparison to market orders in order book data.}.
Let us focus on the third term:
\begin{equation}
    \Delta (\alpha(t) S(t^-)) = \alpha(t+1) \Delta S(t^-) + S(t^-)  \Delta \alpha(t).
\end{equation}
Then
 \begin{eqnarray}
    \mathbb{E} \left[ \left( \Delta (\alpha(t) S(t^-)) \right)^2 \right] & = & \mathbb{E} \left[ (\Delta \alpha(t))^2 \right] \mathbb{E} \left[ S(t^-)^2 \right] \nonumber \\
    & + & 2 \mathbb{E} \left[  \alpha(t+1) \Delta S(t^-)  S(t^-)  \Delta \alpha(t) \right] \nonumber  \\
    & + &    \mathbb{E} \left[ \alpha(t+1)^2  \right] \mathbb{E} \left[   (\Delta S(t^-)) ^2\right].
    \label{IntermediateEq}
\end{eqnarray}
Again, we neglect the cross term\footnote{This time, we are neglecting the correlation between trade signs and spread movements.} in the r.h.s. and we are left with
\begin{eqnarray}
    \mathbb{E} \left[ \left( \Delta (\alpha(t) S(t^-)) \right)^2 \right] & \approx & \mathbb{E}\left[ (\Delta \alpha(t))^2 \right] \mathbb{E} \left[ S(t^-)^2 \right] \nonumber \\
& + &   \mathbb{E} \left[   (\Delta S(t^-)) ^2\right].
\end{eqnarray}
But
\begin{eqnarray}
    \mathbb{E}\left[ (\Delta \alpha(t))^2 \right] & = & \mathbb{E} \left[ \alpha(t+1)^2  \right] + \mathbb{E} \left[ \alpha(t)^2  \right]
- 2 \mathbb{E} \left[ \alpha(t) \alpha(t+1)  \right] \nonumber \\
& = & 2 \left( 1 - \rho_1(\alpha) \right),
\end{eqnarray}
where $\rho_1(\alpha)$ is the autocorrelation of trade signs at the first lag.

Finally\footnote{More generally, after $n$ trades: \begin{equation}
    {\sigma_n^{Tr}}^2 \approx  {\sigma_n}^2 + \frac{1}{2 n} \left( 1 - \rho_n(\alpha) \right)
    {\mathbb{E} \left[ S(t^-)^2 \right]}.
\label{ShortTermDiffGeneral}
\end{equation}}:
\begin{equation}
    {\sigma_1^{Tr}}^2 \approx  {\sigma_1}^2 + \frac{1}{2} \left( 1 - \rho_1(\alpha) \right) \mathbb{E} \left[ S(t^-)^2 \right] + \frac{1}{4} \mathbb{E} \left[   (\Delta S(t^-)) ^2\right].
\label{ShortTermDiff}
\end{equation}

Two effects are clear from equation \eqref{ShortTermDiff}:
\begin{enumerate}
    \item The trade price variance at short time scales is larger than the mid-price variance,
    \item Autocorrelation in trade signs dampens this discrepancy. This partially explains\footnote{Interestingly, although the arguments that led to \eqref{ShortTermDiff} are rather qualitative, a back of the envelope calculation with $\mathbb{E}\left[ S^2 \right] \in [1,9]$, gives a difference ${\sigma^{Tr}}^2 - \sigma^2$ in the range $[0.5, 4.5]$; which has the same order of magnitude of the values obtained by simulation.} why the trades signature plot obtained from the data is flatter than the model predictions: $\rho_1(\alpha)_{model} = 0$, while $\rho_1(\alpha)_{data} \approx 0.6.$
\end{enumerate}

From a modeling perspective, a possible solution to recover the diffusivity even at very short time scales, is to incorporate long-ranged correlation in the order flow. Toth et al. \cite{Toth} have investigated numerically this route using a ``$\epsilon$-intelligence'' order book model. In this model, market orders signs are long-ranged correlated, that is, in trade time
\begin{equation}
\rho_n(\alpha) = \mathbb{E}\left[ \alpha(t+n) \alpha(t) \right] \propto n^{-\gamma}, \qquad \gamma \in ]0, 1[.
\end{equation}
And the size of incoming market orders is a fraction $f$ of the volume displayed at the best opposite quote, with $f$ drawn from the distribution
\begin{equation}
P_{\xi}(f) = \xi (1-f)^{\xi-1},
\end{equation}
They show that, by fine tuning the additional parameters $\gamma$ and $\xi$, one can ensure the diffusive behavior of the price both at mesoscopic ($\approx$ a few trades) and macroscopic ($\approx$ hundred trades) time scales\footnote{Note that Toth. el al. \cite{Toth} model the ``latent order book'', not the actual observable order book. The former represents the \emph{intended} volume at each price level $p$, that is, the volume that would be revealed should the price come close to $p$. So that the interpretation of their parameters, in particular the expected lifetime $\tau_{\text{life}}$ of an order, does not strictly match ours.}.

\section{Conclusions}
\label{Conclusion}

This paper analyzes a simple Markovian order book model, in
which elementary changes in the price and spread processes are explicitly linked
to the instantaneous shape of the order book and the order flow parameters.

Two basic properties were investigated: the ergodicity of the order book and the large-scale limit of the price process. The first property, which we answered positively, is desirable in that it assures the stability of the order book in the long run, and gives a theoretical underpinning to statistical measurements on order book data. The scaling limit of the price process is, as anticipated, a Brownian motion. A key ingredient in this result is the convergence of the order book to its stationary state at an exponential rate, a property equivalent to a geometric mixing condition satisfied by the stationary version of the order book. This short memory effect, plus a constraint on the variance of price increments guarantee a diffusive limit at large time scales. Our assumptions are independent Poissonian order flows, proportional cancellation rates, and the presence of two reservoirs of liquidity $K$ ticks away from the best quotes to guarantee that the spread does not 
diverge.\footnote{We believe this assumption can be relaxed under a balance condition on the arrival rates. One has however to consider an order book model with finite but \emph{ unbounded support}, and control not only the stability of the spread but also of all the gaps in the book.} 

We believe the results hold for a wide class of Markovian order book models: In general, one can state that price increments in a \emph{stable} Markovian order book model are aggregationally Gaussian\footnote{Rigorously, the convergence to the stationary state has to happen \emph{fast enough}. That is, with an integrable convergence rate $\rho(n)$ as in \eqref{integrable}.}.

In a sense, this could offer a mathematical justification to the Bachelier
model of asset prices, from a market microstructure perspective.
In reality, the picture is however more subtle: even if the price process is asymptotically diffusive, at short time scales, the model produces stronger anti-correlation in traded prices than what is actually observed in the data. At those time scales, price diffusivity is arguably the result of a balance between persistent liquidity taking and anti-persistent liquidity providing.

We believe however that the approach presented here is interesting for clearly
 identifying conditions under which the asymptotic normality of price increments holds; and more importantly, for introducing a set of mathematical tools for further investigating the price dynamics in more sophisticated stochastic order book models. Indeed, using the same techniques, we are studying extensions of our results to the case of mutually exciting---and therefore dependent---order flows (point 1 below). This will be published elsewhere.

At this stage of development, our work can naturally be extended in several ways. In the following lines, we suggest some possible avenues to explore. 
	 
	  First of all, actual order flows exhibit non-negligible \emph{cross dependences}. As documented in \cite{MuniToke}, market orders excite limit orders and vice versa. A possible solution for endogenously incorporating these dependences is the use of  mutually exciting processes: 
	\begin{eqnarray}
\lambda^{M}(t) &=& \lambda^{M}(0) +  \int_0^t{ \varphi^{MM}(t-s)} dN^M(s) \nonumber \\ 
&+&  \int_0^t \varphi^{LM}(t-s) dN^L(s),
\end{eqnarray}
and,
\begin{eqnarray}
\lambda^{L}(t) &= &\lambda^{L}(0) +  \int_0^t \varphi^{LL}(t-s) dN^L(s) \nonumber \\
&+&  \int_0^t \varphi^{ML}(t-s) dN^M(s),
\end{eqnarray}
This model has the additional advantage of capturing \emph{clustering} in order arrivals (due to the self-excitation terms $\varphi^{MM}$ and $\varphi^{LL}$), and for exponentially decaying kernels\footnote{$\varphi(u)=\alpha e^{-\beta u}.$} can be cast into a Markovian setting.
	
	 Besides, \emph{long-ranged correlation in order signs} is a very important feature of the data, as discussed in section \ref{NumericalExample}. Analyzing this mathematically is more difficult since the model is no longer Markovian. 

	 Moreover, it is natural to add another source of randomness on the rates themselves, for instance
	\begin{equation}
		d\lambda(t) = \theta (\overline{\lambda}(t) -\lambda(t)) dt + \nu \sqrt{\lambda(t)} dW(t),
	\end{equation}
	where $\overline{\lambda}$ is a (deterministic) background intensity to account for the $U$-shaped daily trading activity and $\theta, \nu$ are the parameters of a CIR process. Such \emph{stochastic arrival rates}  would lead to stochastic volatility in the prices.
	
		Although we argued that the simple Markovian order book model we study is stable and asymptotically diffusive, markets do show signs of fragility quite often and large jumps do occur in actual prices. Understanding how these macroscopic \emph{jumps} (or departure from equilibrium) arise from events at the order book level, for instance via sudden evaporation of liquidity in one side of the book is much needed.
	
	Finally, richer price dynamics (e.g. fat-tailed return distributions) can be obtained using \emph{feedback loops} between the arrival rates and the price (or its volatility) as in \cite{Preis}.

These extensions may, however, render the model less amenable
to mathematical analysis, and we leave the investigation of such interesting
(but sometime difficult) 
questions for future research.

\appendix

\section{Model Parameters}
\label{ModelParameters}
\subsection*{Description of the data}
For reproducibility, we summarize in tables \ref{ModelParams1} and \ref{ModelParams2} the parameters used to obtain figures \ref{AverageDepthProfile}--\ref{SignaturePlot}. These correspond to estimating the model for the stock SCHN.PA (Schneider Electric). Our dataset consists of TRTH\footnote{Thomson Reuters Tick History.} data for the CAC 40 index constituents in March 2011 (23 trading days). We have tick by tick order book data up to $10$ price levels, and trades. A snapshot of these files is given in tables \ref{TickData} and \ref{TradesData}. In order to avoid the diurnal seasonality in trading activity (and the impact of the US market open on European stocks), we somehow arbitrarily restrict our attention to the time window $[9:30$--$14:00]$ CET.

\begin{table}
    \begin{center}
        \begin{tabular}{|c|c|c|c|c|}
                \hline
                Timestamp & Side &  Level &  Price &  Quantity \\
                \hline
                33480.158 & B & 1 & 121.1 & 480 \\
                33480.476 & B & 2 & 121.05 & 1636 \\
                33481.517 & B & 5 & 120.9 & 1318 \\
                \emph{33483.218} & \emph{B} & \emph{1} & \emph{121.1} & \emph{420} \\
                33484.254 & B & 1 & 121.1 & 556 \\
                33486.832 & A & 1 & 121.15 & 187 \\
                33489.014 & B & 2 & 121.05 & 1397 \\
                \emph{33490.473} & \emph{B} & \emph{1} & \emph{121.1} & \emph{342} \\
                \emph{33490.473} & \emph{B} & \emph{1} & \emph{121.1} & \emph{304} \\
                \emph{33490.473} & \emph{B} & \emph{1} & \emph{121.1} & \emph{256} \\
                33490.473 & A & 1 & 121.15 & 237 \\
                \hline
                \end{tabular}
        \caption{\label{TickData} Tick by tick data file sample. Note that the field ``Level'' does not necessarily correspond to the distance in ticks from the best opposite quote as there might be gaps in the book. Lines corresponding to the trades in table \ref{TradesData} are highlighted in \emph{italics}. }
    \end{center}
\end{table}

\begin{table}
    \begin{center}
        \begin{tabular}{|c|c|c|}
        \hline
                    Timestamp & Last & Last quantity \\
                    \hline
                    33483.097 & 121.1 & 60 \\
                    33490.380 & 121.1 & 214 \\
                    33490.380 & 121.1 & 38 \\
                    33490.380 & 121.1 & 48 \\
                    \hline
            \end{tabular}
            \caption{\label{TradesData}Trades data file sample.}
    \end{center}
\end{table}

\subsection*{Trades and tick by tick data processing}
Because one cannot distinguish market orders from cancellations in tick by tick data, and since the timestamps of the trades and tick by tick data files are asynchronous, we use a matching procedure to reconstruct the order book events. In a nutshell, we proceed as follows for each stock and each trading day:
\begin{enumerate}
    \item Parse the tick by tick data file to compute order book state variations:
     \begin{itemize}
        \item If the variation is positive (volume at one or more price levels has increased), then label the event as a limit order.
        \item If the variation is negative (volume at one or more price levels has decreased), then label the event as a ``likely market order''.
        \item If no variation---this happens when there is just a renumbering in the field ``Level'' that does not affect the state of the book---do not count an event.
    \end{itemize}
    \item Parse the trades file and for each trade:
        \begin{enumerate}
            \item Compare the trade price and volume to likely market orders whose timestamps are in $[t^{Tr}-\Delta t, t^{Tr}+\Delta t]$, where $t^{Tr}$ is the trade timestamp and $\Delta t$ is a predefined time window\footnote{We set $\Delta t = 3 \; \text{s}$ for CAC 40 stocks. We found that the median reporting delay for trades is $-900 \; \text{ms}$: on average, trades are reported 900 milliseconds \emph{before} the change is recorded in tick by tick data.}.
            \item Match the trade to the first likely market order with the same price and volume and label the corresponding event as a market order---making sure the change in order book state happens at the best price limits.
            \item Remaining negative variations are labeled as cancellations.
        \end{enumerate}
\end{enumerate}
Doing so, we have an average matching rate of around $85\%$ for CAC 40 stocks. As a byproduct, one gets the sign of each matched trade, that is, whether it is buyer or seller initiated.

\subsection*{Parameters estimation}
If $T$ be the trading duration of interest each day ($T=4.5$ hours---$[9:30$--$14:00]$---in our case.) Then
\begin{equation}
    \widehat{\lambda^M} := \frac{\#\text{trades}}{2 T},
\end{equation}
and
\begin{multline}
    \widehat{\lambda^L_i} := \frac{1}{2 T} .\\
    \left( \#\text{\small buy limit orders arriving $i$ tick away from the best opposite quote} \right. \\
    + \left. \#\text{\small sell lim. orders etc.} \right).
\end{multline}
For cancellations, we need to normalize the count by the average number of shares $\left<\mathbf{X}_i\right>$ at distance $i$ from the best opposite quote:
\begin{multline}
    \widehat{\lambda^C_i} := \frac{1}{\left<\mathbf{X}_i\right>} \frac{1}{2 T} .\\
    \left( \#\text{\small cancellation orders in the bid side arriving $i$ tick away from the best opposite quote} \right. \\
    + \left. \#\text{\small cancellation orders in the ask side etc.} \right),
\end{multline}
We then average $\widehat{\lambda^M}$, $\widehat{\lambda^L_i}$ and $\widehat{\lambda^L_i}$ across 23 trading days to get the final estimates.
As for the volumes, we estimate by maximum likelihood the parameters $(\widehat{v}, \widehat{s})$ of a lognormal distribution separately for each order type. We depict the parameters in figures \ref{RatesPanel} and \ref{VolumePanel}.

\begin{table}
    \begin{center}
        \begin{tabular}{|c|c|}
            \hline
            K & 30\\
            \hline \hline
            $a_\infty$ & 250 \\
            $b_\infty$ & 250 \\
            \hline \hline
            $(v^M, s^M)$ & $(4.00, 1.19)$ \\
            $(v^L, s^L)$ & $(4.47, 0.83)$ \\
            $(v^C, s^C)$ & $(4.48, 0.82)$ \\
            \hline \hline
            $\lambda^{M^{\pm}}$ & 0.1237\\
            \hline
        \end{tabular}
        \caption{\label{ModelParams1}Model parameters for the stock SCHN.PA (Schneider Electric) in March 2011 (23 trading days). Figures \ref{RatesPanel} and \ref{VolumePanel} are graphical representation of these parameters.}
    \end{center}
\end{table}

\begin{table}
    \begin{center}
        \begin{tabular}{|c|c|c|c|}
            \hline
            $i$ (ticks) & $\left<\mathbf{X}_i\right>$ (shares) & $\lambda^{L^{\pm}}_{i}$ & $10^3 . \lambda^{C^{\pm}}_{i}$\\
            \hline
            1   &   276 &   0.2842  &   0.8636  \\
            2   &   1129    &   0.5255  &   0.4635  \\
            3   &   1896    &   0.2971  &   0.1487  \\
            4   &   1924    &   0.2307  &   0.1096  \\
            5   &   1951    &   0.0826  &   0.0402  \\
            6   &   1966    &   0.0682  &   0.0341  \\
            7   &   1873    &   0.0631  &   0.0311  \\
            8   &   1786    &   0.0481  &   0.0237  \\
            9   &   1752    &   0.0462  &   0.0233  \\
            10  &   1691    &   0.0321  &   0.0178  \\
            11  &   1558    &   0.0178  &   0.0127  \\
            12  &   1435    &   0.0015  &   0.0012  \\
            13  &   1338    &   0.0001  &   0.0001  \\
            14  &   1238    &   0.0       &   0.0       \\
            15  &   1122    &   \vdots  &   \vdots  \\
            16  &   1036    &           &      \\
            17  &   943     &           &      \\
            18  &   850     &           &      \\
            19  &   796     &           &      \\
            20  &   716     &           &      \\
            21  &   667     &           &      \\
            22  &   621     &           &      \\
            23  &   560     &           &      \\
            24  &   490     &           &     \\
            25  &   443     &           &      \\
            26  &   400     &           &      \\
            27  &   357     &           &      \\
            28  &   317     &           &      \\
            29  &   285     &   \vdots  &   \vdots   \\
            30  &   249     &   0.0       &   0.0       \\
            \hline
        \end{tabular}
        \caption{\label{ModelParams2}Model parameters for the stock SCHN.PA (Schneider Electric) in March 2011 (23 trading days). Figures \ref{RatesPanel} and \ref{VolumePanel} are graphical representation of these parameters.}
    \end{center}
\end{table}

\begin{figure}
\begin{center}
\includegraphics[width=0.95\textwidth]{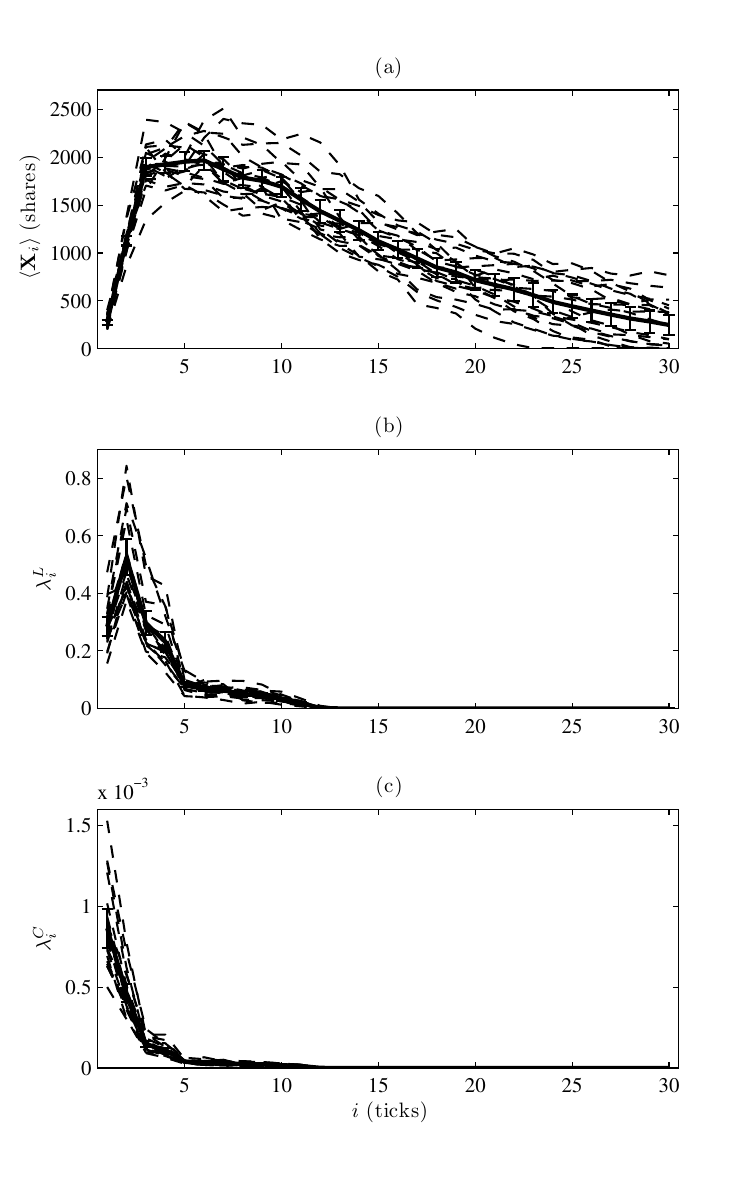}
\caption{\label{RatesPanel}Model parameters: arrival rates and average depth profile (parameters as in table \ref{ModelParams2}). Error bars indicate variability across different trading days.}
\end{center}
\end{figure}

\begin{figure}
\begin{center}
\includegraphics[width=0.95\textwidth]{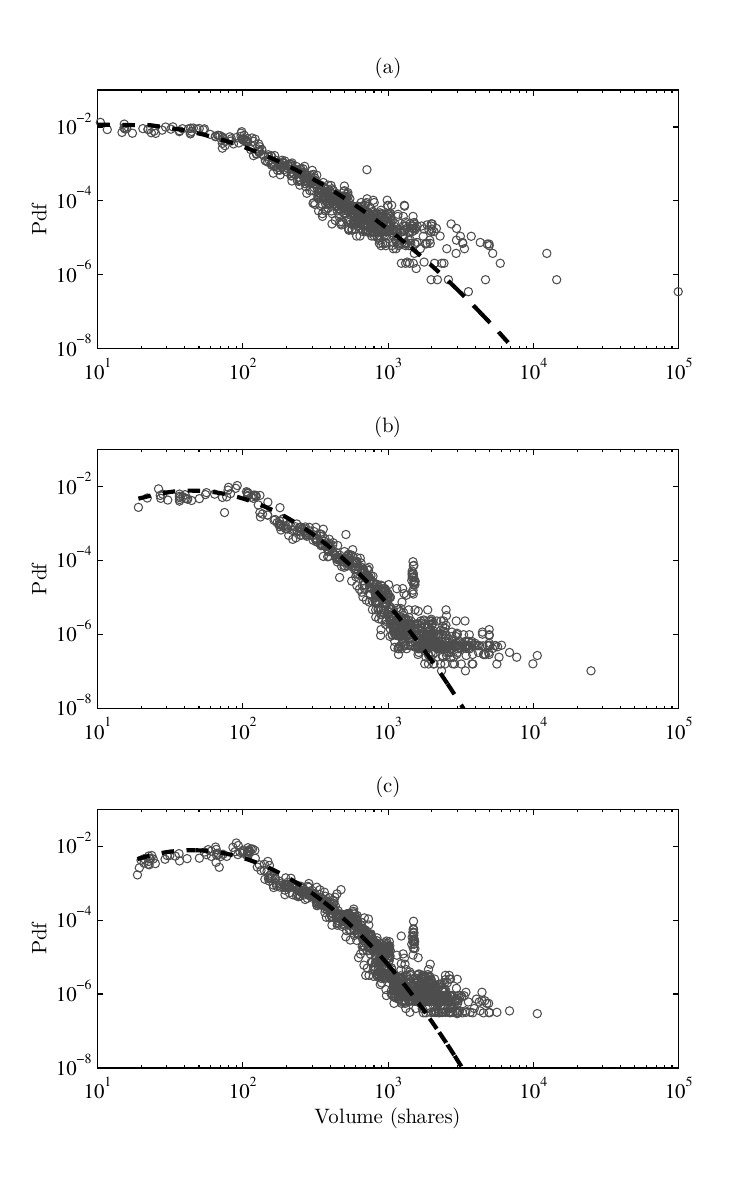}
\caption{\label{VolumePanel}Model parameters: volume distribution. Panels $(a), (b)$ and $(c)$ correspond respectively to market, limit and cancellation orders volumes. Dashed lines are lognormal fits (parameters as in table~\ref{ModelParams1}).}
\end{center}
\end{figure}

\section{Results for CAC 40 stocks}

In order to get a cross-sectional view of the performance of the model on all CAC 40 stocks, we estimate the parameters separately for each stock and run a $100,000$ event simulation for each parameter set. We then compare in figure \ref{Cac40Panel} the average depth, average spread and the long-term ``volatility'' measured directly from the data, to those obtained from the simulations. Dashed line is the identity function---It would correspond to a perfect match between model predictions and the data. Solid line is a linear regression $z_{\text{data}}~=~b_1~+~b_2~z_{\text{model}}$ for each quantity of interest $z$.

Note that despite the good agreement between the average depth profiles (panel $(a)$), and although the model successfully predicts the relative magnitudes of the long-term variance $\sigma_{\infty}^2$ and the average spread $\left<S \right>$ for different stocks, it tends to systematically underestimate $\sigma_{\infty}^2$ and $\left<S \right>$. This may be related to the absence of autocorrelation in order signs in the model and the presence of more drifting phases in actual prices than in those obtained by simulation.

\begin{figure}
\begin{center}
\includegraphics[width=0.95\textwidth]{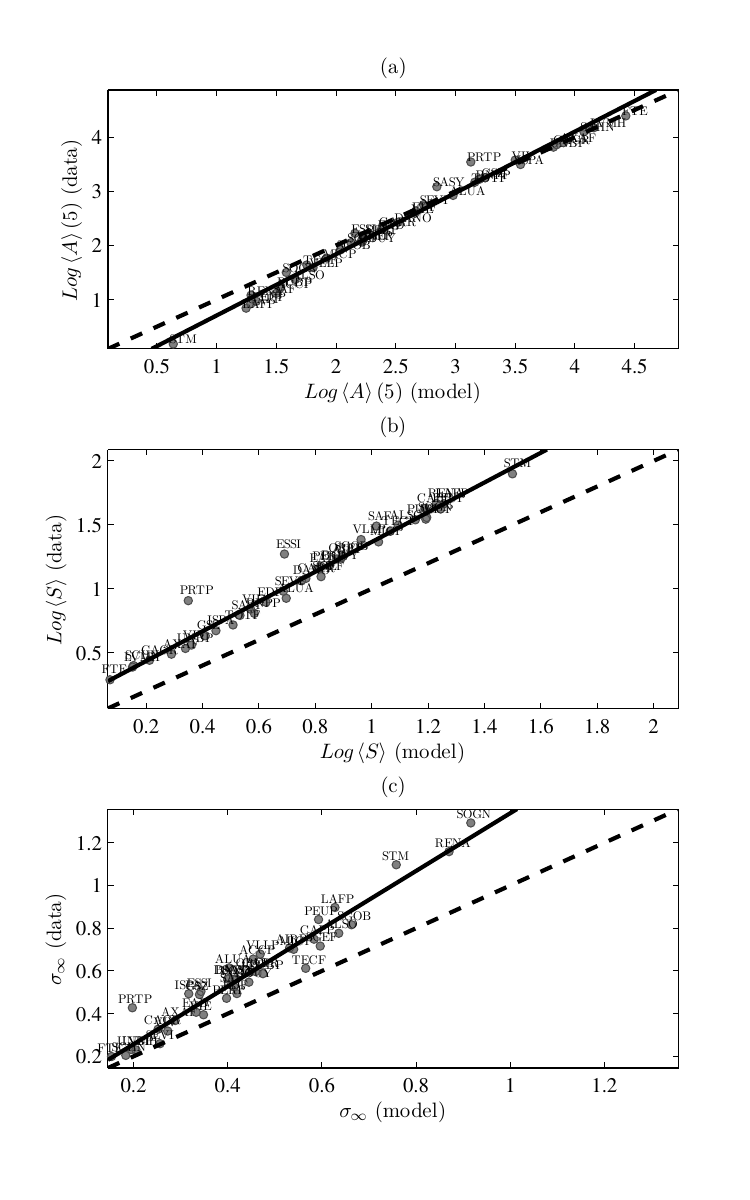}
\caption{\label{Cac40Panel}A cross-sectional comparison of liquidity and price diffusion characteristics between the model and data for CAC 40 stocks (March 2011).}
\end{center}
\end{figure}

\begin{table}
    \begin{center}
        \begin{tabular}{|c|c|c|c|}
            \hline
             & $b_1$ & $b_2$ & $R^2$ \\
            \hline
            $Log \left<A\right>(5)$& $-0.42 \, (\pm 0.11)$ & $1.13 \, (\pm 0.04)$ & $0.99$\\
            \hline
            $Log \left<S\right>$& $ 0.20 \, (\pm 0.06)$ & $1.16 \, (\pm 0.07)$ &  $0.97$\\
            \hline
            $\sigma_{\infty}$& $-0.012 \, (\pm 0.05)$ & $1.35 \, (\pm 0.11)$ & $0.94$\\
            \hline
        \end{tabular}
        \caption{\label{Regression}CAC 40 stocks regression results.}
    \end{center}
\end{table}

\section*{Acknowledgments}

We warmly thank D. Challet, N. Millot, O. Torne and R. Zaatour
for useful discussions and two anonymous referees for very helpful
comments that improved this paper considerably and for pointing out reference \cite{Toth}.


\end{document}